\documentclass[10pt,journal,twoside]{IEEEtran}
\usepackage{cite}
\usepackage{graphicx}
\usepackage{psfrag}
\usepackage{amsthm}
\usepackage{url}
\usepackage{amsmath}
\usepackage{array}
\usepackage{amssymb}
\usepackage{authblk}
\usepackage{amsfonts}
\usepackage{graphicx}
\usepackage{epstopdf}
\usepackage{algorithm}

\newtheorem{lemma}{Lemma}

\newtheorem{corollary}{Corollary}

\newtheorem{theorem}{Theorem}

\newtheorem{example}{Example}

\usepackage{setspace}
\usepackage{amscd}
\usepackage{amsthm}
\usepackage{mathrsfs}
\usepackage{epsfig}
\usepackage{color}
\usepackage{textcomp}
\usepackage{multirow}
\usepackage{caption}
\usepackage{threeparttable}

\title{On The Number of Optimal Linear Index Codes For Unicast Index Coding Problems}

\begin{document}
\author{Kavitha R., Niranjana Ambadi and B. Sundar Rajan
}\affil{Dept. of ECE, IISc, Bangalore 560012, India, Email: kavithar1991@gmail.com, bsrajan@ece.iisc.ernet.in.}
\date{\today}
\renewcommand\Authands{ and }
\maketitle
\thispagestyle{empty}	
\begin{abstract}
An index coding problem arises when there is a single source with a number of  messages and multiple receivers each wanting a subset of messages and knowing a different set of messages a priori. The noiseless Index Coding Problem is to identify the minimum number of transmissions (optimal length) to be made by the source through noiseless channels so that all receivers can decode their wanted messages using the transmitted symbols and their respective prior information. Recently \cite{TKCR}, it is shown that different optimal length codes perform differently in a noisy channel. Towards identifying the best optimal length index code one needs to know the number of  optimal length index codes. In this paper we present results on the number of optimal length index codes making  use of the representation of an index coding problem by an equivalent network code. Our formulation results in matrices of smaller sizes compared to the approach of Kotter and Medard \cite{KoM}. Our formulation leads to a lower bound on the minimum number of optimal length codes possible for all unicast index coding problems \cite{OnH} which is met with equality for several special cases of the unicast index coding problem. A method to identify the optimal length codes which lead to minimum-maximum probability of error is also presented.\footnote{Part of the content of this manuscript appears in \cite{KaR}}
\end{abstract}


\section{Introduction}
\IEEEPARstart{W}{e} consider the index coding problem first introduced by Birk and Kol in \cite{BiK}. In an index coding (IC) problem, there is a single source with a set of  messages and a set of  receivers,  each  wanting  a subset of messages and knowing a subset of messages (side-information) a priori.   A general index coding problem can be formulated as follows: There are $n$ messages, $ x_1, x_2, \ldots, x_n $ and $m$ receivers. Each receiver wants a subset of messages, $ W_i $ and knows a subset of messages $ K_i. $ For a  general unicast problem, $ W_i \cap  W_j =\emptyset$, for $ i \neq j. $ 

A single uniprior IC problem is a scenario where each receiver knows a single unique message (not known to other receivers) a priori  and a unicast problem is one  where each receiver wants a unique set of messages, i.e., the intersection of the messages wanted by any two receivers is nullset. The general scenario is called group-cast IC problem.  A single unicast is when the size of each of those wanted sets in a unicast problem is one. Without loss of generality a unicast problem can be reduced to a single unicast problem by increasing the number of receivers by splitting any receiver with more than one message into multiple receivers wanting only one message with identical side information. One needs to identify the minimum number of transmissions to be made so that all receivers can decode their wanted messages using the transmitted bits and their respective prior information. Ong and Ho in \cite{OnH} gave an algorithm which finds  the optimal length of a uniprior index coding problem. When a general unicast IC problem is modified to a single unicast IC problem one has $n=m.$ In this paper we consider single unicast IC  problem but the results  apply to a general unicast problem as well.

El Rouayheb \textit{et. al.} in [4] found that every index coding problem can be reduced to an equivalent network coding problem. An algebraic representation of network codes in terms of matrices representing the input mixing, topology and the output mixing operations  was given by Koetter and Medard in \cite{KoM}. In this paper we present a similar  algebraic characterization of a single unicast IC problem after reducing it to an equivalent network code. Harvey \textit{et.al} in \cite{HKM} proposed an algorithm for network codes for multicast problems, which is based on a new algorithm for maximum-rank completion of mixed matrices. Our problem is not a multicast problem. Hence the results in \cite{HKM} cannot be applied.

For a given index code a receiver may not use all the transmissions from the source. In fact, different receivers may use different number of transmissions of the source. It has been shown in \cite{TKCR} that there can be several linear optimal index codes  in terms of lowest number of transmissions for an IC problem, but among them one needs to identify the linear optimal index code which minimizes the maximum number of transmissions that is required by any receiver in decoding its desired message. The motivation for this comes from the fact that each of the transmitted symbols is error prone in a  wireless scenario and lesser the number of transmissions used in decoding the desired message, lesser will be its probability of error. Hence among all the codes with the same length, the one for which the maximum number of transmissions used by any receiver is the minimum, will have minimum-maximum error probability. This has already been discussed in \cite{TKCR} for single uniprior IC problems where a method to find a best linear solution in terms of minimum-maximum error probability among all codes with the optimal length is given.

The contributions  of this paper may be summarized as follows:
\begin{itemize}
\item A transfer matrix approach similar to that of Kotter and Medard \cite{KoM}, but with component matrices of much smaller sizes, is presented which enables identification of the optimal length for any unicast problem in terms of the component matrices of the transfer matrix.   
\item A lower bound on the number of optimal length codes for any single unicast index coding problem is obtained. This bound is shown to be exact for the following special cases:\\
(i) Single uniprior single unicast IC problems, \\
(ii) Single uniprior unicast IC problems and \\ 
(iii) Single unicast uniprior IC problems.
\item For  single uniprior unicast problems and single unicast uniprior problems,  we obtain the length of optimal linear ICs by reducing the problem to that of single uniprior single unicast IC problem. 
\item A criterion for optimal linear index codes with minimum-maximum probability of error is presented in terms of the component matrices for any single unicast problem.
\end{itemize}

The remaining content is organized as follows: In Section \ref{sec2}  the equivalence of index coding problem to network coding problem is discussed and in Section \ref{sec3} we obtain the input-mixing matrix, transfer matrix and the output-mixing matrices for the index coding problem and show that they can be partitioned into submatrices corresponding to the side-information and the index code used.  Bounds on the number of optimal linear index codes is presented in Section \ref{sec4} and the number of codes with optimal length is discussed in Section \ref{sec5}. In Section \ref{sec6} a method to identify optimal codes with minimum-maximum error probability is described and simulation results are presented.

\section{Equivalent network coding problem}
\label{sec2}
Throughout the paper we assume that the  operations are over the finite field with two elements $\mathbb{F}_{2}$. But the results  easily  carry over to other finite field.
\begin{figure}[htbp]
\centering
\includegraphics[scale=.75]{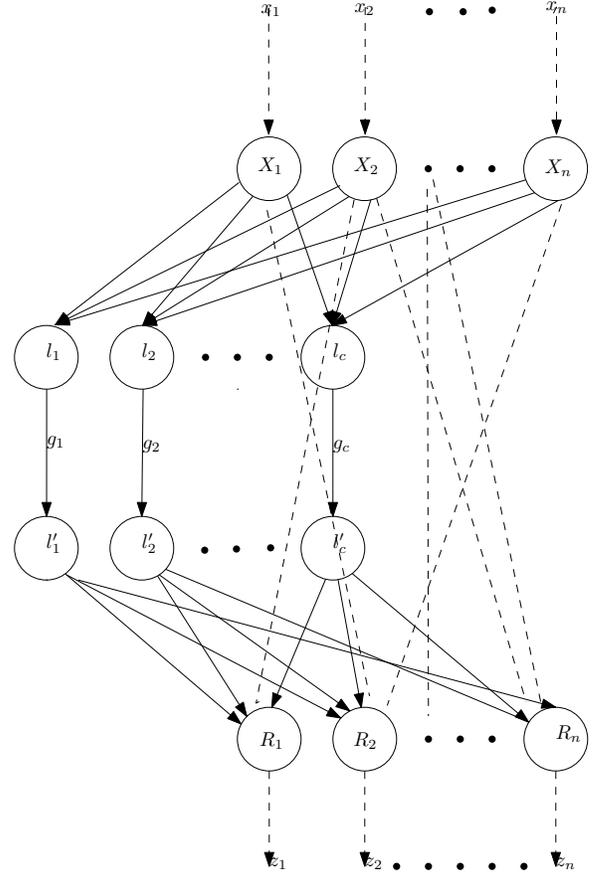}
\caption{Representation of a unicast IC problem by an equivalent network code.}
\label{fig1}
\end{figure}
Any single unicast problem can be represented by an equivalent network coding problem as in Fig. \ref{fig1}. This was proposed by El Rouayheb \textit{et. al.} in \cite{RSG}. Suppose the length of the index code, not necessarily optimal, be $c.$ Each of the messages $ x_1, x_2,\ldots, x_n $ is represented by a source node and $ g_1, g_2,\ldots, g_c $ represent the  broadcast channel and $l_{1},l_{2},\ldots,l_{c},l'_{1},l'_{2},\ldots,l'_{c}$ represent the intermediate  nodes. It is assumed that when two or more edges have the same tail node with only one incoming edge to it they carry the same message. Hence,  $l'_{i}$ transmits to its outgoing edges whatever it gets through $g_{i}$. The source nodes $X_i$ transmit their respective messages as such through their outgoing edges. The minimum possible  value of $c$ among all linear solutions of an IC problem is the optimal length. The dashed lines represent the connection between a receiver node and a source message node whose message is known to the receiver apriori,  i.e, they represent the side information possessed by the receivers. For every single unicast problem, we can find a graph like given in Fig. \ref{fig1}, which we will denote as $G$ henceforth. The graph $G$ can be represented as $G=(V, E)$, where $V=\lbrace x_{1}, x_{2}, \ldots, x_{n}, l_{1}, l_{2}, \ldots, l_{c}, l'_{1}, l'_{2}, \ldots, l'_{c}, R_{1}, R_{2}, \ldots, R_{n} \rbrace$ is the vertex set and $E$ is the edge set. It is easy to see that  $\mid E \mid$ = $(2n+1)c + \overset{n}{\underset{i=1}{\sum}}\mid K_{i} \mid $. An edge connecting vertex $ v_{1}$ to $v_{2}$ is denoted by ($v_{1},v_{2}$) where $v_{1}$ is the tail of the edge and $v_{2}$ is the head of the edge. For an edge $e$, $Y(e)$ represents the bit passing through  that edge. We can get a transfer matrix $ M_{n \times n} $ (which is shown in Section \ref{sec2}) such that $ \b{Z} = [   z_{1} ~z_{2}~ \ldots ~z_{n}  ]^T, $ the vector of output messages from all the receivers, can be expressed as 
     \begin{equation}
\label{eqn1}
     Z=  M~X,
     \end{equation}
      where ${X}$ = $[ x_{1} ~  x_{2} ~\ldots~  x_{n}]^T$, the vector of input messages. Hence, the considered index code of length $c$ gives a  solution  in $c$ number of transmissions if $M$ is the identity matrix.
\begin{figure*}
\begin{eqnarray}
\b{Y}^T=[ Y((x_{1},l_{1}))~ Y((x_{1},l_{2}))~ \ldots ~ Y((x_{1},l_{c}))\notag\\  Y((x_{2},l_{1}))~ Y((x_{2},l_{2})) ~ \ldots ~ Y((x_{2},l_{c}))\notag\\ \vdots \notag\\  Y((x_{n},l_{1}))~ Y((x_{n},l_{2}))~\ldots ~  Y((x_{n},l_{c}))\notag\\  Y((x_{K_{1,1}},R_{1}))~Y((x_{K_{1,2}},R_{1}))\ldots~ Y((x_{K_{1,\mid K_{1}\mid}},R_{1}))\notag\\ Y((x_{K_{2,1}},R_{2}))~Y((x_{K_{2,2}},R_{2}))~\ldots ~Y((x_{K_{2,\mid K_{2}\mid}},R_{2}))\notag \\  \vdots \notag \\ Y((x_{K_{n,1}},R_{n}))~Y((x_{K_{n,2}},R_{n}))~Y((x_{K_{n,\mid K_{n}\mid}},R_{n}))] 
\label{eqn:YY}
\end{eqnarray}
\hrule
\end{figure*}
~ \\

\section{Transfer function matrices for an index code}
\label{sec3}

For a general single unicast problem, we can find a matrix $M_{n\times n}$ in \eqref{eqn1} such that  $M$ is a product of three matrices as \footnote{We are not following Koetter and Medard's approach \cite{KoM}. The approach in \cite{KoM} would have given  matrix $A$ of order ($\mid E \mid \times n$), $F$ of order ($\mid E \mid \times \mid E \mid$) and $B$ of order ($n \times\mid E \mid $) where $\mid E \mid$ = $(2n+1)c + \overset{n}{\underset{i=1}{\sum}}\mid K_{i} \mid $. Our formulation results in  matrices $A$, $F$ and $B$  for a given index coding problem, where $\mid E \mid$ = $nc + \overset{n}{\underset{i=1}{\sum}}\mid K_{i} \mid $. Notice that our matrices $A,$ $F$ and $B$ are of much smaller sizes compared to the sizes one deals with by following the approach in \cite{KoM}.} 
\begin{equation}
M=B~F~ A
\label{eqn:M}
\end{equation}
\noindent where the matrix $A$ relates the input messages and the messages flowing through the outgoing edges of all the source nodes, the matrix $F$ relates to the messages sent in the broadcast channel and the side information possessed by the the receivers, and  the matrix $B$ describes the decoding operations done at the receivers. All these three matrices can be partitioned in to two parts one corresponding to the index codeword transmission and the other corresponding to only the side information. Now we proceed to describe these partitioning. \\

\noindent {\bf Partitioning of $A$:} The matrix $A$ satisfies the relation
\begin{equation}
 {Y}=  A ~ {X},\\
 \label{eqn:Y}
\end{equation}
\begin{figure*}
 \begin{eqnarray}
 \b{Y}'^T= [Y((l'_{1},R_{1}))~Y((l'_{1},R_{2}))~ \ldots Y((l'_{1},R_{n}))\notag\\Y((l'_{2},R_{1}))~Y((l'_{2},R_{2}))~\ldots Y((l'_{2},R_{n}))\notag\\\vdots \notag\\ Y((l'_{c},R_{1})) ~Y((l'_{c},R_{2}))~\ldots ~Y((l'_{c},R_{n})) \notag \\Y((x_{K_{1,1}},R_{1}))~Y((x_{K_{1,2}},R_{1}))\ldots Y((x_{K_{1,\mid K_{1}\mid}},R_{1}))\notag\\ ~Y((x_{K_{2,1}},R_{2}))~Y((x_{K_{2,2}},R_{2})) ~\ldots Y((x_{K_{2,\mid K_{2}\mid}},R_{2})) \notag \\ \vdots \notag \\ Y((x_{K_{n,1}},R_{n}))~Y((x_{K_{n,2}},R_{n}))\ldots Y((x_{K_{n,\mid K_{n}\mid}},R_{n})) ]
  \label{eqn:YYY}
  \end{eqnarray}
  \vspace{-.3 cm}
\end{figure*}
  \begin{figure*}
\hrule
\vspace{.4 cm}
\begin{equation} 
\begin{footnotesize}
F_{BC}=\left[ \begin{array}{cccccccccccccc}
\beta_{X_{1},l_{1}}&0& \ldots &0&\beta_{X_{2},l_{1}}&0& \ldots &0&\ldots &\beta_{X_{n},l_{1}}&0& \ldots &0 \\
\beta_{X_{1},l_{1}}&0& \ldots &0&\beta_{X_{2},l_{1}}&0& \ldots &0&\ldots &\beta_{X_{n},l_{1}}&0& \ldots &0 \\
.\\
.\\
\beta_{X_{1},l_{1}}&0& \ldots &0&\beta_{X_{2},l_{1}}&0& \ldots &0&\ldots &\beta_{X_{n},l_{1}}&0& \ldots &0 \\
0&\beta_{X_{1},l_{2}}& \ldots &0&0&\beta_{X_{2},l_{2}}& \ldots &0&\ldots &0&\beta_{X_{n},l_{2}}& \ldots &0\\
0&\beta_{X_{1},l_{2}}& \ldots &0&0&\beta_{X_{2},l_{2}}& \ldots &0&\ldots &0&\beta_{X_{n},l_{2}}& \ldots &0\\
.\\
.\\

0&\beta_{X_{1},l_{2}}& \ldots &0&0&\beta_{X_{2},l_{2}}& \ldots &0&\ldots &0&\beta_{X_{n},l_{2}}& \ldots &0\\
.\\
.\\
.\\
0&0& \ldots &\beta_{X_{1},l_{c}}&0&0& \ldots &\beta_{X_{2},l_{c}}&\ldots &0&0& \ldots&\beta_{X_{n},l_{c}}\\
0&0& \ldots &\beta_{X_{1},l_{c}}&0&0& \ldots &\beta_{X_{2},l_{c}}&\ldots &0&0& \ldots &\beta_{X_{n},l_{c}}\\
.\\
.\\
0&0& \ldots &\beta_{X_{1},l_{c}}&0&0& \ldots &\beta_{X_{2},l_{c}}&\ldots &0&0& \ldots &\beta_{X_{n},l_{c}}\end{array} \right]
 \end{footnotesize}
  \label{eqn:fsub}
 \end{equation}
 \hrule
 \end{figure*}
\noindent
where  ${Y}^T$ is as in (\ref{eqn:YY}), with $K_{i,j}$ denoting the index of $j$-th message in the side information set of receiver $R_{i}$ and ${X} =[ x_{1} ~ x_{2} ~ x_{3}\ldots ~x_{n}]^T$ is the vector of input messages. The vector $\b{Y}$ is the vector of messages flowing through the outgoing edges of all the source nodes and is of order $((nc+\overset{n}{\underset{i=1}{\sum}} \mid K_{i} \mid)\times 1 )$. The matrix $A$ is of order $ (nc+\overset{n}{\underset{i=1}{\sum}} \mid K_{i} \mid) \times n $ and it can be split in the form,
 \begin{eqnarray}
 A=\left[ \begin{array}{l}
 A_{BC} \\
 A_{SI}\\
 \end{array} \right]
 \end{eqnarray}
\noindent where $ A_{BC} $  is of order $nc \times n$ and $A_{SI}$ is of order $\overset{n}{\underset{i=1}{\sum}}\mid K_{i} \mid \times n$ (the subscript $BC$ stands for "broad cast part" and the subscript $SI$ for "side information part"). The matrix $A_{BC}$ is a matrix formed by row-concatenation of matrices $A_{i}$, $i=1,\ldots n$ where each $A_{i}$ is a $c \times n$ matrix in which all elements in the $i$-th column are ones and the rest all are zeros as given in (\ref{ll}).

{\footnotesize
 \begin{equation}
  \begin{array}{r@{}}
    \text{$A_{1}$}~\left\{\begin{array}{ccccccc}\null\\\null\\\null\\\null\\\null\end{array}\right.\\
     \text{$A_{2}$}~\left\{\begin{array}{ccccccc}\null\\\null\\\null\\\null\\\null\end{array}\right.\\
      \text{$.$}\\
      \text{$.$}\\
       \text{$A_{n}$}~\left\{\begin{array}{ccccccc}\null\\\null\\\null\\\null\\\null\end{array}\right.
  \end{array}
  \left [
   \begin{array}{cccccccccccccc}
1& 0 & 0 & 0& \ldots &0\\
1& 0 & 0 & 0& \ldots&0\\
.\\
.\\
1& 0 & 0 & 0& \ldots &0\\
0 &1& 0&0& \ldots &0 \\
0 &1& 0&0 & \ldots &0 \\
.\\
.\\
0 &1& 0&0& \ldots &0 \\
.\\
.\\
0 &0 &0 &0& \ldots &1\\
0 &0 &0 &0& \ldots &1\\
.\\
.\\
0 &0 &0 &0& \ldots &1
   \label{ll}     
    \end{array}
  \right ]
\end{equation}
}
\noindent 
Each $A_{i}$ corresponds to the message passed by the source node $x_{i}$ to the intermediate nodes, $l_{j}$, $j=1,\ldots,c$. 

The matrix  $A_{SI}$ has only one non-zero element (which is one) in each row. This matrix corresponds to the side information possessed by the receivers and each successive set of $\mid K_{i} \mid$ rows correspond to the side information possessed by $ R_{i} $ for $i=1$ to $n$. In each set of $ \mid K_{i} \mid $ rows, each row is distinct and has only one non-zero element  which occupies the respective column position of one of the messages in the prior set of $R_{i}$. 

Notice that  matrix $A$ is fixed for a fixed $c$ and does not depend on the index code except on its length. \\

\begin{figure*}
  \hrule
~ \\
  \begin{equation}
  \begin{footnotesize}
  B_{BC}=\left[ \begin{array}{ccccccccccccccccc}
\epsilon_{l'_{1},R_{1}}&0&0& \ldots &0&\epsilon_{l'_{2},R_{1}}&0&0& \ldots &0& \ldots &\epsilon_{l'_{c},R_{1}}&0&0& \ldots &0\\
0&\epsilon_{l'_{1},R_{2}}&0& \ldots &0&0&\epsilon_{l'_{2},R_{2}}&0& \ldots &0& \ldots &0&\epsilon_{l'_{c},R_{2}}&0& \ldots &0\\
 ....\\
 .....\\
 .....\\
 0&0&0& \ldots &\epsilon_{l'_{1},R_{n}}& 0&0&0& \ldots &\epsilon_{l'_{2},R_{n}}& \ldots &0&0&0& \ldots &\epsilon_{l'_{c},R_{n}}
 \end{array} \right]
 \end{footnotesize}
 \label{eqn:Bsub}
  \end{equation}
\hrule
 \end{figure*} 
\noindent {\bf Partition of the matrix $F$:} The matrix $F$ which relates the message symbols to the transmission symbols sent in the broadcast channel and the side information possessed by the the receivers is of order $(nc+\overset{n}{\underset{i=1}{\sum}}\mid K_{i} \mid)\times  (nc+\overset{n}{\underset{i=1}{\sum}}\mid  K_{i} \mid)$. It satisfies
  \begin{equation}
   {Y}' =F~{Y}=F~A~{X},
   \label{eqn:Y'}
   \end{equation}
\noindent  where ${Y}'^T$ is as in (\ref{eqn:YYY}), and is the vector of messages flowing to each of the receiver. We can observe that $F$ can be split into four block matrices as  
  \begin{eqnarray}
                F=\left[ \begin{array}{cc}
 F_{BC} & 0\\
  0  &I \end{array} \right].
  \label{eq:F}
  \end{eqnarray}
\noindent The Matrix $F_{BC}$ is a square matrix of order $nc$ which is of the form given in (\ref{eqn:fsub}) and $I$ is the identity matrix of size $\sum_{i=1}^n |K|_i.$ The elements $\beta_{X_{i},l_{j}},\forall i=1,\dots,n$ and $j=1,\ldots,c$ belong to $\mathbb{F}_{2}$. All the  $((i-1)n+1$)-th to $((i-1)n+n)$-th row are identical for $i=1,2,\ldots,c$. If  any one of these rows is denoted by  $t_{i}$, then we have the $i-$th transmission symbol of the code given by
  \begin{equation}
\label{gsubi}
  g_i=t_{i}~A_{BC}~\b{X}=\sum_{j=1}^{n} \beta_{X_j,l_i}x_j
  \end{equation}
  for $i=1,2,\ldots,c$. Notice that we are using $g_i$ to denote both the $i-$th symbol of the index codeword as well as the edge carrying that symbol. \\

 \begin{figure*}
   \hrule
   \vspace{.2 cm}
\begin{equation}
 F_{BC}=\left[ \begin{array}{cccccccccccccccccccc}
\beta_{X_{1},l_{1}}&0& \beta_{X_{2},l_{1}}&0 & \beta_{X_{3},l_{1}}&0  \\
\beta_{X_{1},l_{1}}&0& \beta_{X_{2},l_{1}}&0 & \beta_{X_{3},l_{1}}&0  \\
\beta_{X_{1},l_{1}}&0& \beta_{X_{2},l_{1}}&0 & \beta_{X_{3},l_{1}}&0  \\
0&\beta_{X_{1},l_{2}}&  0&\beta_{X_{2},l_{2}} &0&\beta_{X_{3},l_{2}}  \\
0&\beta_{X_{1},l_{2}}&  0&\beta_{X_{2},l_{2}} &0&\beta_{X_{3},l_{2}}  \\
0&\beta_{X_{1},l_{2}}&  0&\beta_{X_{2},l_{2}} &0&\beta_{X_{3},l_{2}}  
\end{array} \right]  
\label{eqn:example1fsub}
\end{equation}
 \hrule
\end{figure*}
\begin{figure*}
\begin{equation}
  B=\left[ \begin{array}{cccccccccccccc}
\epsilon_{(l'_{1},R_{1})}&0&0&\epsilon_{(l'_{2},R_{1})}&0&0 & \epsilon_{(x_{2},R_{1})}&0&0\\
0&\epsilon_{(l'_{1},R_{2})}&0&0&\epsilon_{(l'_{2},R_{2})}&0&0& \epsilon_{(x_{3},R_{2})}&0\\
0&0&\epsilon_{(l'_{1},R_{3})}&0&0&\epsilon_{(l'_{2},R_{3})} &0&0& \epsilon_{(x_{1},R_{3})}
 \end{array} \right]
 \label{bexample1}
\end{equation}
 \hrule
\end{figure*}
\noindent {\bf Partitioning of matrix $B$:}  The matrix $B$ that describes the decoding operations done at the receivers is of order $n \times (nc+\overset{n}{\underset{i=1}{\sum}}\mid K_{i} \mid)$. It satisfies the relation,
  \begin{equation}
   {Z} =B~{Y}'=B~F~A~{X}.
   \label{eqn:B}
   \end{equation}
In terms of the symbols, 
\begin{equation}
\begin{array}{cl}
z_j &= \sum_{i=1}^c \epsilon_{l'_i,R_j} g_i \\
    &= \sum_{i=1}^c \epsilon_{l'_i,R_j} \left( \sum_{k=1}^n \beta_{x_k,l_i} x_k  \right) \\
    &= \sum_{i=1}^c \sum_{k=1}^n \epsilon_{l'_i,R_j}  \beta_{x_k,l_i} x_k 
\end{array}
\end{equation}
gives the symbols obtained by the receivers after all the operations.\\
\noindent The  matrix $B$ can be split into two block matrices as 
\begin{equation}
  B =\left[ \begin{array}{cc}
 B_{BC} & B_{SI}  \end{array} \right] ,
 \end{equation}
where  $ B_{BC}$ is a matrix of order $ n \times nc$ and in every row only $c$ elements are non-zero and the non-zero elements correspond to whether or not $R_{i}$ uses that particular transmission to decode its wanted message. The matrix $ B_{SI}$ is of order $ n \times  \overset{n}{\underset{i=1}{\sum}}\mid K_{i} \mid$. It relates to the side information possessed by the receivers. In this matrix all elements except the $i$-th element in every successive set of $\mid  K_{i}\mid$ columns are zeros, for all $i=1$ to $n$. The rest of the elements are either one or zero  and it depends on the messages used by a receiver  to decode its wanted message. The matrix $B_{B}$ is as in (\ref{eqn:Bsub}). The elements $\epsilon_{l'_{j},R_{i}}$ for $j=1,\ldots, c$ and $i=1,\ldots,n$ belong to $\mathbb{F}_{2}$. From (\ref{eqn:Y}), (\ref{eqn:Y'}) and (\ref{eqn:B}), we get\\
 \begin{equation}
   {Z}  =B~F~A~{X}.
   \label{eqn:relation}
   \end{equation}
   So,
    \begin{equation}
   M  = B~F~A.
   \label{eqn:relationm}
   \end{equation}
 An index coding problem  is solvable with $c$ number of transmissions if  we can find variables $\beta$'s in $F_{BC}$ which define the code and the variables $\epsilon$'s in $B_{BC}$ which define the decoding operations  such that $M$ is an identity matrix.

The following example illustrates the partitioning of the matrices $A,$ $F$ and $B.$
\begin{figure}[htbp]
\centering
\includegraphics[scale=.75]{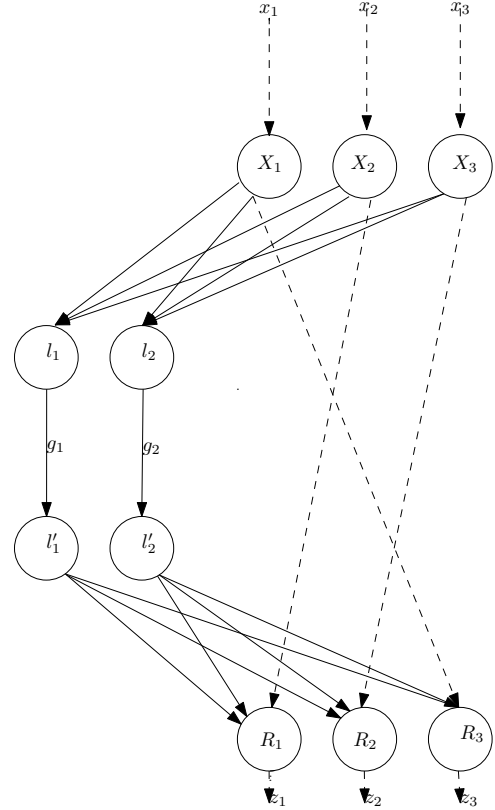}
\caption{ Equivalent network code corresponding to the IC problem in Example \ref{ex1}}
\end{figure}

\begin{example}
\label{ex1}
Let $ m=n=3$. Each $R_{i}$ wants $x_{i}$ and knows $x_{i+1}$, where $+$ is mod-3 addition. The optimal length of a linear IC solution for this problem is $2,$  which is established in {\it Example 1 (continued)} Section \ref{sec4} by way of showing that length one codes do not exist. The graph G for $c=2$ is as in Fig. 2.

The $A$ matrix is as  below.
\begin{eqnarray}
A = \left[ \begin{array}{ccc}
1 & 0 & 0\\
1& 0 & 0\\
0 & 1 & 0 \\
0 & 1 & 0 \\
0 & 0  &1 \\
0 & 0  &1 \\
0 & 1 & 0 \\
0 & 0  &1\\
1& 0 & 0
 \end{array} \right] 
\end{eqnarray}
The last three rows of the above matrix is $A_{SI}.$ The matrix $F_{BC}$ is as in (\ref{eqn:example1fsub}) and $B$ matrix is as in (\ref{bexample1}). There are three  linear codes which are optimal in terms of length. They are 
$$\mathfrak{C}_{1}= \left\{ x_{1}   \oplus   x_{2} ,  x_{2}   \oplus   x_{3}\right\},$$  
$$\mathfrak{C}_{2}= \left\{x_{1}   \oplus   x_{3},  x_{3}   \oplus   x_{2} \right\} ,$$   
$$\mathfrak{C}_{3}= \left\{ x_{1}  \oplus   x_{3} ,  x_{1}   \oplus  x_{2} \right\}.$$ 
For the codes  $\mathfrak{C}_{1}$,  $\mathfrak{C}_{2}$ and $\mathfrak{C}_{3},$ the matrices $F_{BC}$ and $B$ are as in (\ref{eqn:example1Bsubone}), (\ref{eqn:example1Bsubone1}) and (\ref{eqn:example1Bsubone2}) respectively. The matrix $B_{SI}$ is given by 
\begin{eqnarray}
B_{SI} = \left[ \begin{array}{ccc}
x & 0 & 0\\
0& x & 0\\
0 & 0 & x 
 \end{array} \right] 
\end{eqnarray}
\noindent
where $x$ is either $0$ or $1$.
\end{example}
\begin{figure*}
{\small 
\begin{eqnarray}
 F_{BC}=\left( \begin{array}{cccccccccccccccccccc}
1&0& 1 &0 & 0&0  \\
1&0& 1 &0 & 0&0  \\
1&0& 1 &0 & 0&0  \\
0&0&  0&1 &0&1  \\
0&0&  0&1 &0&1  \\
0&0&  0&1 &0&1  
\end{array} \right),
& B=\left( \begin{array}{cccccccccccccc}
1&0&0&0&0&0 & 1 &0&0\\
0&0&0&0&1&0&0& 1&0\\
0&0&1&0&0&1 &0&0& 1
\end{array} \right)
\label{eqn:example1Bsubone}
\end{eqnarray}
\begin{eqnarray}
F_{BC}=\left( \begin{array}{cccccccccccccccccccc}
1&0& 0 &0 & 1&0  \\
1&0& 0 &0 & 1&0  \\
1&0& 0 &0 & 1&0  \\
0&0&  0&1 &0&1  \\
0&0&  0&1 &0&1  \\
0&0&  0&1 &0&1  
\end{array} \right),
& B=\left( \begin{array}{cccccccccccccc}
1&0&0&1&0&0 & 1 &0&0\\
0&0&0&0&1&0&0& 1&0\\
0&0&1&0&0&0 &0&0& 1
\end{array} \right)
\label{eqn:example1Bsubone1}
\end{eqnarray}
\begin{eqnarray}
F_{BC}=\left( \begin{array}{cccccccccccccccccccc}
1&0& 0 &0 & 1&0  \\
1&0& 0 &0 & 1&0  \\
1&0& 0 &0 & 1&0  \\
0&1&  0&1 &0&0  \\
0&1&  0&1 &0&0  \\
0&1&  0&1 &0&0  
\end{array} \right),
& B=\left( \begin{array}{cccccccccccccc}
0&0&0&1&0&0 & 1 &0&0\\
0&1&0&0&1&0&0& 1&0\\
0&0&1&0&0&0 &0&0& 1
\end{array} \right)
\label{eqn:example1Bsubone2}
\end{eqnarray}
}
\hrule
\end{figure*}
\section{Bounds on the number of optimal linear ICs}
 \label{sec4}
We have analyzed the structures of the three matrices in the previous section. We need $M =  B~F~A$ to be $I_n,$ the identity matrix. Here for a fixed length $c$, $A $ is fixed and as can be verified all the columns of $A $ are independent and hence the rank of $A$ is $n$. So, the rows of $I_n$ lie in the row space of $A$. Therefore, the equation
$$T_{n \times (nc+\overset{n}{\underset{i=1}{\sum}}\mid K_{i}\mid)}A=I_n $$
\noindent 
has at least one solution for the variables of the matrix  $T$ which is a left-inverse of $A.$ Observe that the number of free variables  in $T$ is  $( {n^{2}c-n^{2}+n\overset{n}{\underset{i=1}{\sum}}\mid K_{i}\mid})$ and the number of pivot variables is $n^{2}$ \cite{Str}. Hence the number of left inverses of $A$ is $ 2^{n^{2}c-n^{2}+n\overset{n}{\underset{i=1}{\sum}}\mid K_{i}\mid} $. We need to find a  matrix $T$ which is a left inverse of  $A$  as well as is a product $BF$ of some  $B$ and $F$ in the required form. Let us denote  the set of all left inverses of $A$ which are of the form $BF$ in the required form, by  $S(c)$ since it is a function of $c.$ Note that  $B$ and $F$ are needed to be in the form that includes the constraints imposed by the side-information of the index coding problem.  Since 
\begin{equation}
BF=
\left[ \begin{array}{cc}
B_{BC} B_{SI} 
\end{array} \right] 
\left[ \begin{array}{cc}
F_{BC} & 0 \\
 0  &I
\end{array} \right] 
=
\left[  \begin{array}{c}
 B_{BC}F_{BC} ~~B_{SI}
\end{array} \right],
\end{equation} 
\noindent let $T= \left[ \begin{array}{cc} T_{BC} T_{SI}  \end{array} \right].$  This gives 
\begin{equation}
T _{SI} =  B_{SI}.
\label{addiconstraint}
\end{equation}
So the positions which are to be  occupied by zeros in $ B_{SI}$ are needed to be zeros in $ T_{SI} $ also. Therefore, $ T_{SI} $ which is of order $n  \times \overset{n}{\underset{i=1}{\sum}}\mid  K_{i} \mid$  has $(n-1)(\overset{n}{\underset{i=1}{\sum}} \mid  K_{i} \mid)$ zeroes and  when the rest of the elements of $T_{SI}$ are fixed, $ B_{SI}$ also gets fixed. Let the left inverses of $A$ that satisfy the constraint \eqref{addiconstraint} be denoted by  $S^\prime(c)$. Clearly $S(c) \subseteq   S^\prime(c) $.  As the rank of  $A$  is $n$, the total number of  left inverses of  $A$  with restrictions given by \eqref{addiconstraint} is 
\begin{equation}
\label{sizeofSdash}
|S^\prime (c)|=2^{n^{2}c-n^{2}+\overset{n}{\underset{i=1}{\sum}}\mid K_{i}\mid}.
\end{equation}
Since $S(c) \subseteq   S'(c) $, we have,
\begin{equation}
  \vert  S(c)  \vert   \leq  2^{n^{2}c-n^{2}+\overset{n}{\underset{i=1}{\sum}} \mid K_{i}\mid}.
  \label{23}
 \end{equation}
 We need to identify the elements in the set $S'(c)$ which also belong to $S(c)$. First of all, when we fix $T$, the submatrix $B_{SI}$ gets fixed. So, for a pair $(B,F)$ which belongs to set $S(c)$, we have
  \begin{equation}
 B_{BC} F_{BC}=T_{BC}.
 \label{FF} 
   \end{equation}          
 From (\ref{FF}) we get relations of the form,
 \begin{eqnarray}
 \label{eqn:columncond}
 \left[ \begin{array}{c}
 \epsilon_{l'_{i},R_{1}}\\
 .\\
 .\\
 \epsilon_{l'_{i},R_{n}}\end{array} \right]  
    \begin{array}{c}
 \beta_{X_{k},l'_{i}}
\end{array}   =\left[ \begin{array}{c}
 T_{col'_{(k-1)c+i}}
\end{array} \right]
\end{eqnarray} 
$ \forall   k \in \lbrace1,2,...n\rbrace $ and $\forall i \in \lbrace1,2,....c\rbrace $
where $ T_{col_{i}}$ is the $i$-th column of $T_{BC}$.

For $i=1,2,...,c,$ we define the $n \times n$ matrix  ${\cal R}_i$ as  consisting of the $n$ columns  $ \lbrace   T_{col_{i}},   T_{col_{c+i}} ... T_{col_{(n-1)c+i}}  \rbrace $ of the matrix $T_{BC}.$ Notice that the matrix ${\cal R}_i$ consists of transmissions to all the receivers  from the $i-$th transmission of the index code. Also note that
 \begin{eqnarray}
 \label{Rsubi}
\begin{array}{c}
{\cal R}_i
\end{array}
=
\left[ \begin{array}{c}
 \epsilon_{l'_{i},R_{1}}\\
 .\\
 .\\
 \epsilon_{l'_{i},R_{n}}\end{array} 
\right]  
\left[
    \begin{array}{cccccc}
 \beta_{X_{1},l_{i}}~~
\beta_{X_{2},l_{i}}~~
. . .~~
\beta_{X_{n},l_{i}}

\end{array}   
\right]
\end{eqnarray}
and 
{\footnotesize
 \begin{eqnarray}
 \label{Rsub2}
{\cal R}_i
\left[
\begin{array}{c}
x_1 \\
x_2 \\
. \\
. \\
. \\
x_n
\end{array}
\right]
\hspace{-0.1cm}
=
\hspace{-0.1cm}
\left[ \begin{array}{c}
 \epsilon_{l'_{i},R_{1}}\\
 .\\
 .\\
 \epsilon_{l'_{i},R_{n}}\end{array} 
\right]  
\hspace{-0.1cm}
\left[
    \begin{array}{cccccc}
 \beta_{X_{1},l_{i}}~
\beta_{X_{2},l_{i}}~
. . .~
\beta_{X_{n},l_{i}}

\end{array}   
\right]
\hspace{-0.1cm}
\left[
\begin{array}{c}
x_1 \\
x_2 \\
. \\
. \\
. \\
x_n
\end{array}
\right].
\end{eqnarray}
}
\noindent 
Equation \eqref{Rsub2} above shows that the contribution of the $i-$th transmission of an index code is completely captured by the matrix ${\cal R}_i.$

\begin{lemma}
Any matrix  $T$ which belongs to $S^\prime(c)$ also belongs to $S(c)$  if and only if ${\cal R}_i$ is a all-zero matrix or is a rank one matrix, for all $i=1,2,...,c.$
\end{lemma} 
\begin{proof}  (Only-if part): If $T\in S(c)$, then from \eqref{Rsubi}, it follows that either  ${\cal R}_i$ is a all-zero matrix or its rank is one. \\
\noindent Proof of 'if part' : For a $T \in S^\prime (c),$ for all $i,$ one can  always find values for variables $\epsilon$'s and $\beta$'s  satisfying \eqref{Rsubi}. Hence one can get a pair $(B,F)$ such that (\ref{FF}) is satisfied by substituting these values. Hence $T\in S(c)$. This completes the proof. 
\end{proof}
 \begin{figure*}
 \vspace{-.5 cm}
  \begin{eqnarray}
\left[ \begin{array}{cccccccccc}
1 & 0 &0 &\ldots &p_{ {1, \lbrace j :K_{j}=x_{1}\rbrace}} & 0 &\ldots &0\\
 0  &  1 & 0& \ldots &p_{{2,\{ j' :K_{j'}=x_{2} \}}} & 0& \ldots &0\\
 .\\
 .\\
 0 &0 &0 &\ldots &p_{{n,\{ j'' :K_{j''}=x_{n}\}}} & 0&\ldots &1
\end{array} \right]
 \label{eqn:minnum3}
\end{eqnarray}
\hrule
\end{figure*}

However for a $T \in S(c)$, if ${\cal R}_i$ is a zero matrix, then either all the $\beta $'s or $ \epsilon $'s corresponding to ${\cal R}_i$ in \eqref{Rsubi} are zeros. When all the $\beta$'s are zeros, the $\epsilon$'s can take any of the $2^n$ values possible and vice verse. Hence the number of possibilities for a zero matrix is $2^{n+1}-1$. Hence the total number of ($B, F $) possible for a $T$ matrix is $(2^{n+1}-1)^{\lambda }$, where $\lambda $, $0 \leq \lambda \leq c$ is the number of zero matrices among ${\cal R}_i$'s for $i=1,2,...,c.$  
\begin{theorem} A length $c$ is optimal for a linear index coding problem if and only if for all the matrices $T \in S(c)$ none of the  $c$ number of ${\cal R}_i$ matrices is a all-zero matrix. In other words, $\lambda=0$ for every $T$ in $S(c).$ 
\end{theorem}
\begin{proof} 
Proof for 'only if' part: We need to prove that if there exists a $T \in S(c) $ whose $\lambda \neq 0$ for a particular length $c$, then $c$ is not the optimal transmission length. Since $\lambda \neq 0,$ let  $1 \leq i \leq c$ be such that ${\cal R}_i$ is the all-zero matrix. Then we have  either all the $\beta$'s or all the  $ \epsilon $'s corresponding to ${\cal R}_i$ are zeros. If all the $ \epsilon $ are zeroes, that means that  one particular transmission is not  used by any of the receivers. Else, if all the $\beta$'s corresponding are  zeroes, then no message is transmitted in one particular transmission. In either case,  we can remove at least one transmission. Hence $c$ is not the optimal length. \\
\noindent The proof for 'if part' is by contradiction. Assume that a length $c$ exists such that it is feasible but not optimal and all the matrices in $S(c)$ have $\lambda=0$. Assume further that $c'= c-r $ for some $r > 0$, is  the optimal length. Then take one feasible solution with length $c'$ and add  extra $nr$ rows to the corresponding $F_{B}$ matrix and some extra $nc$ all zero columns to $B_{B}$. Let us call the new matrices $F'_{B}$ and $B'_{B}$. Let $g'_{i}, i=1,\ldots c$ be the set of broadcast messages given by $F'_{B}$ and $g_{i}$ be those which are given by $F_{B}$. One can observe that $\lbrace g'_{1},g'_{2},\ldots, g'_{c}\rbrace$ is nothing but $\lbrace g_{1},g_{2},\ldots,g_{c'}\rbrace$ plus some additional information. Hence when one sends  $\lbrace g'_{1},g'_{2},\ldots, g'_{c}\rbrace$, the receivers get whatever they would have got if $\lbrace g_{1},g_{2},\ldots,g_{c'}\rbrace$ was sent. Hence even if they do not use the extra transmissions given by $F'_{B}$, they will be able to decode their wanted messages. Hence the product of $F'_{B}$ and $B'_{B}$ matrices should belong to $S(c)$ (as it is a feasible index code) and has $\lambda \neq 0$, which is a contradiction.  Hence $c$ is the optimal length.
\end{proof}

Theorem 1 is illustrated in Example-1 (continued) and Example-2 below.

\textbf{Example 1.} \textit{(continued)}. We  illustrate Theorem 1 for the problem in Example 1. We will prove $c=1$ is not possible in this case. With $n=3$ and $c=1,$ from \eqref{23} there can be at most are $2^{12}$ matrices in $S(c).$ The matrix $B_{SI}$ is of the form 
$\left[\begin{array}{ccc}       
x & 0 & 0 \\
0 & x & 0 \\
0 & 0 & x \\
\end{array}   
\right]$   
where $x$ can be either $0$ or $1.$  From \eqref{sizeofSdash}, we have $2^3=8 $ matrices that belong to $S^\prime(1).$ We found these by brute force among $2^{12}$ matrices which has zeros at places which are occupied by zeros strictly in the corresponding $B_{SI}$. These eight matrices are  given below.
 \[
\left[ \begin{array}{ccc|ccc}
1  & 1 & 0 &1 &0 &0 \\
0   &1  &0  &0 &0 &0 \\
0  &0  & 1 &0 &0 & 0
 \end{array} 
\right],
\left[ \begin{array}{ccc|ccc}
 1 &1  & 0 &1 &0 &0 \\
0   &1  &0  &0 &0 &0 \\
1  & 0 & 1 &0 &0 & 1
 \end{array} 
\right]
\]
 \[
\left[ \begin{array}{ccc|ccc}
1  & 1 & 0 &1 &0 &0 \\
0   &1  &1  &0 &1 &0 \\
0  &0  & 1 &0 &0 &0 
 \end{array} 
\right],
\left[ \begin{array}{ccc|ccc}
1  & 1 & 0 &1 &0 &0 \\
0   &1  &1  &0 &1 &0 \\
1  & 0 & 1 &0 &0 & 1
 \end{array} 
\right]
\]
 \[
\left[ \begin{array}{ccc|ccc}
1  &0  &0  &0 &0 & 0\\
0   &1  &0  &0 &0 &0 \\
0  &0  &1  &0 &0 & 0
 \end{array} 
\right],
\left[ \begin{array}{ccc|ccc}
1  &0  &0  &0 &0 &0 \\
0   &1  &0  &0 &0 &0 \\
1  &0  &1  &0 &0 & 1
 \end{array} 
\right]
\]
 \[
\left[ \begin{array}{ccc|ccc}
1  &0  &0  &0 &0 &0 \\
0   & 1 &1  &0 &1 &0 \\
0  &0  &1  &0 &0 & 0
 \end{array} 
\right],
\left[ \begin{array}{ccc|ccc}
1  &0  &0  &0 &0 &0 \\
0   &1  &1  &0 &1 &0 \\
1  & 0 &1  &0 &0 & 1
 \end{array} 
\right]
\]
\noindent Since $c=1$ there is only one ${\cal R}_1$ which is the $3 \times 3$ submatrix consisting of the first 3 columns in each of the 8 matrices above. Clearly none of these matrices have rank one. Hence, there does not exist a solution with $c=1$.
\begin{example}
\label{ex2}
Let $m=n=3$ and  $ R_{i} $ wants  $ x_{i} $, $\forall i \in \lbrace{1,2,3}\rbrace$. $ R_{1} $ knows $ x_{2} $ and $ x_{3} $. $ R_{2} $ knows $ x_{3} $. $ R_{3} $ knows $ x_{1} $. We will show that the optimal value of $c$ is $2$. For $c=1$, size of $S^\prime(c)=16$ (from  \eqref{sizeofSdash}. The 16 matrices  which belong to $S^\prime(1)$ have been found by brute force among $2^{13}$ matrices which has zeros at places, which are to be occupied strictly by zeros in the corresponding $B_{SI}=
\left[ \begin{array}{cccc}
x & x & 0 & 0 \\
0 & 0 & x & 0 \\
0 & 0 & 0 & x
\end{array}
\right]
$ where $x$ stands for either $0$ or $1.$ These 16 matrices are shown below.
\[
\left[ \begin{array}{ccc|cccc}
1 &0 &0 &0 &0 &0 &0 \\
0 &1 &0 &0 &0 &0 &0 \\
0 &0 &1 &0 &0 &0 &0
 \end{array} 
\right],
\left[ \begin{array}{ccc|cccc}
1 &0 &1 &0 &1 &0 &0\\
0 &1 &0 &0 &0 &0 &0\\
0 &0 &1 &0 &0 &0 &0
 \end{array} 
\right]
\]
\[
\left[ \begin{array}{ccc|cccc}
1 &1 &0 &1 &0 &0 &0 \\
0 &1 &0 &0 &0 &0 &0 \\
0 &0 &1 &0 &0 &0 &0
 \end{array} 
\right],
\left[ \begin{array}{ccc|cccc}
1 &1 &1 &1 &0 &0 &0\\
0 &1 &0 &0 &0 &0 &0\\
0 &0 &1 &0 &0 &0 &0
 \end{array} 
\right]
\]
\[
\left[ \begin{array}{ccc|cccc}
1 &0 &0 &0 &0 &0 &0 \\
1 &1 &1 &0 &0 &0 &0 \\
1 &0 &1 &0 &0 &0 &1
 \end{array} 
\right],
\left[ \begin{array}{ccc|cccc}
1 &1 &0 &1 &0 &0 &0 \\
1 &1 &1 &0 &0 &1 &0 \\
1 &0 &1 &0 &0 &0 &1
 \end{array} 
\right]
\]
\[
\left[ \begin{array}{ccc|cccc}
1 &0 &1 &0 &1 &0 &0 \\
1 &1 &1 &0 &0 &1 &0 \\
1 &0 &1 &0 &0 &0 &1
 \end{array} 
\right],
\left[ \begin{array}{ccc|cccc}
1 &1 &1 &1 &1 &0 &0\\
1 &1 &1 &0 &0 &1 &0\\
1 &0 &1 &0 &0 &0 &1
 \end{array} 
\right]
\]
\[
\left[ \begin{array}{ccc|cccc}
1 &0 &0 &0 &0 &0 &0 \\
0 &1 &1 &0 &0 &1 &0 \\
0 &0 &1 &0 &0 &0 &0
 \end{array} 
\right],
\left[ \begin{array}{ccc|cccc}
1 &1 &0 &1 &0 &0 &0 \\
0 &1 &1 &0 &0 &1 &0 \\
0 &0 &1 &0 &0 &0 &0
 \end{array} 
\right]
\]
\[
\left[ \begin{array}{ccc|cccc}
1 &0 &1 &0 &1 &0 &0 \\
0 &1 &1 &0 &0 &1 &0 \\
0 &0 &1 &0 &0 &0 &0
 \end{array} 
\right],
\left[ \begin{array}{ccc|cccc}
1 &1 &1 &1 &1 &0 &0 \\
0 &1 &1 &0 &0 &1 &0 \\
0 &0 &1 &0 &0 &0 &0
 \end{array} 
\right]
\]
\[
\left[ \begin{array}{ccc|cccc}
1 &0 &0 &0 &0 &0 &0 \\
0 &1 &0 &0 &0 &0 &0 \\
1 &0 &1 &0 &0 &0 &1
 \end{array} 
\right],
\left[ \begin{array}{ccc|cccc}
1 &1 &0 &1 &0 &0 &0 \\
0 &1 &0 &0 &0 &0 &0 \\
1 &0 &1 &0 &0 &0 &1
 \end{array} 
\right]
\]
\[
\left[ \begin{array}{ccc|cccc}
1 &0 &1 &0 &1 &0 &0 \\
0 &1 &0 &0 &0 &0 &0 \\
1 &0 &1 &0 &0 &0 &1
 \end{array} 
\right],
\left[ \begin{array}{ccc|cccc}
1 &1 &1 &1 &1 &0 &0 \\
0 &1 &0 &0 &0 &0 &0 \\
1 &0 &1 &0 &0 &0 &1
 \end{array} 
\right]
\]
Since $c=1,$ there is only one ${\cal R}_1$ for each of the 16 matrices above and this ${\cal R}_1$ is the $3 \times 3$ submatrix consisting of the first 3 columns. It is easily seen that all these matrices have rank more than 1. Hence $c =1$ is not a feasible length for this case. For the case $c=3,$ the following  matrix $T$  belongs to the set $S(3).$
  \begin{eqnarray}
T=
\left[ \begin{array}{ccccccccc|cccc}
1 &0 &0 &1 &0 &0 &0 &0 &0 &1 &0 &0 &0 \\
0 &0 &0 &0 &1 &0 &0 &1 &0 &0 &1 &0 &0 \\
1 &0 &0 &1 &1 &0 &0 &1 &0 &0 &0 &0 &1 
  \end{array} \right]
  \label{ls}
  \end{eqnarray}
For this matrix $T,$ we have 
$$ 
\small{ 
{\cal R}_1 
\hspace{-0.1cm}
= 
\hspace{-0.1cm}
\left[ \begin{array}{ccc}
1 & 1 & 0 \\
0 & 0 & 0 \\
1 & 1 & 0
\end{array} \right],   
{\cal R}_2 
\hspace{-0.1cm}
=
\hspace{-0.1cm}
 \left[ \begin{array}{ccc}
0 & 0 & 0 \\
0 & 1 & 1 \\
0 & 1 & 1
\end{array} \right],   
{\cal R}_3 
\hspace{-0.1cm}
=
\hspace{-0.1cm}
 \left[ \begin{array}{ccc}
0 & 0 & 0 \\
0 & 0 & 0 \\
0 & 0 & 0
\end{array} \right]. 
}
$$
We see that  $\lambda =1$ due to the all-zero matrix ${\cal R}_3$ and also the rank of the other two matrices is one. Hence $c=3$ is not optimal. Therefore, $c=2 $ should be the optimal length. 
\end{example}
\section{Minimum Number of Codes Possible for an Optimal $c$}
\label{sec5}
  In this subsection, we  find a lower bound on the number of linear codes  which are optimal in terms of length or equivalently in terms of bandwidth, for a single unicast index coding problem. For the optimal $c$, the number of  matrices which are left inverses of $A$ and is a product of  some $B$ and $F$ gives the number of codes possible with that length, which is also the size of  the set $S(c)$. But for any $T \in S(c)$, 
\begin{equation}
T_{BC}A_{BC}=B_{BC}F_{BC}A_{BC} =  I-T_{SI} A_{SI}=I-B_{SI}A_{SI}
\label{eqn:minnum1}
\end{equation}\\
which is equal to (\ref{eqn:minnum2}).
\begin{figure*}
 \hrule
\begin{eqnarray}
\left( \begin{array}{cccccc}
\epsilon _{l'_{1},R_{1}} & \epsilon _{l'_{2},R_{1}} &\ldots &\epsilon _{l'_{c},R_{1}} \\
\epsilon _{l'_{1},R_{2}} & \epsilon _{l'_{2},R_{2}} &\ldots &\epsilon _{l'_{c},R_{2}} \\
 . \\
 .\\
 .\\
\epsilon _{l'_{1},R_{n}} & \epsilon _{l'_{2},R_{n}} &\ldots&\epsilon _{l'_{c},R_{n}}
\end{array} \right)
 \left( \begin{array}{cccccccc}
\beta_{X_{1},l_{1}} & \beta_{X_{2},l_{1}} &\ldots &\beta_{X_{n},l_{1}}\\
\beta_{X_{1},l_{2}} & \beta_{X_{2},l_{2}}  &\ldots &\beta_{X_{n},l_{2}}\\
 . \\
 .\\
 .\\
\beta_{X_{1},l_{c}} & \beta_{X_{2},l_{c}} &\ldots &\beta_{X_{n},l_{c}}
\end{array} \right)  
 \label{eqn:minnum2}
\end{eqnarray}
\hrule
\end{figure*}
\begin{theorem}
\label{thm2} 
The number of linear index coding solutions having optimal length $c$ for a single unicast IC problem is at least
 \begin{equation}
  \frac{1}{c!} \prod_{i=0}^{c-1}{(2^{c}-2^{i})}
  \label{eqn:minnum4}
  \end{equation}
\end{theorem}
\begin{proof}:
 Consider (\ref{eqn:minnum1}) and (\ref{eqn:minnum2}). Here if both RHS of (\ref{eqn:minnum1})  and the second matrix in (\ref{eqn:minnum2}) are fixed, a solution which is the first matrix in (\ref{eqn:minnum2})  will exist only if the row space of RHS of (\ref{eqn:minnum1}) is spanned by the rows of the second matrix in  (\ref{eqn:minnum2}). But the rank of the second matrix in (\ref{eqn:minnum2}) is at most $c$. Hence this is possible only if the rank of the RHS matrix in  (\ref{eqn:minnum1}) is less than or equal to $c$. The number of possible  $ B_{SI}$ matrices is $2^{{\overset{n}{\underset{i=1}{\sum}}} \mid  K_{i}\mid}$. Since  $c$ is the optimal length, there should be at least one $ B_{SI}$ such that RHS of (\ref{eqn:minnum1}) is of rank $c$. For any such RHS of (\ref{eqn:minnum1}), we can take the second matrix in (\ref{eqn:minnum2}) in 

{\small 
$$(2^{c}-1){\overset{c-1}{\underset{i=1}{\prod}}} \left( 2^{c}-1 -\dbinom{i}{1}-\dbinom{i}{2}.......-\dbinom{i}{i}\right)= \prod_{i=0}^{c-1} (2^c-2^i) $$
}
\noindent
ways such that the row spaces of both the matrices are same. Each such matrix is an index code, which is feasible, and each row of the  matrix represents a transmission. As order of transmissions does not matter, we need to neglect those matrices which are row-permuted versions of one another. Hence, total number of distinct transmission schemes possible is  
${\frac{1}{c!}    {\overset{c-1}{\underset{i=0}{\prod}}{(2^{c}-2^{i})} }}$. But there may be more than one $B_{SI}$ matrices which are of rank $c$ and whose row spaces are different. Hence the total number of index codes possible can be more than (\ref{eqn:minnum4})  as we take into account all possible basis sets of each of the different row spaces. (Example 3 is such a case.)  Hence (\ref{eqn:minnum4}) is a lower bound on the number of index codes possible.
\end{proof}

Note that all possible matrices occupying RHS of (\ref{eqn:minnum1}) are exactly the collection of matrices which fits the index coding problem as per the definition of a fitting matrix in \cite{YBJK}. Hence algebraically we have proved the already established result in \cite{YBJK} that the optimal length of a linear solution is the minimum among the ranks of all the matrices which fits the IC problem. However, we will see subsequently that the matrices $A_{SI}$ and $I-B_{SI}A_{SI}$ will be useful in finding the optimal length and the number of optimal linear ICs in some special cases much more easily than using fitting matrices.  Moreover, combining equations (\ref{gsubi}), (\ref{eqn:minnum1}) and (\ref{eqn:minnum2}), we see that the different bases of the $I-B_{SI}A_{SI}$ matrix can be used to obtain all possible linear optimal ICs. The elements of each such basis are precisely the codewords of the code determined by the chosen basis. This is illustrated in detail in Example \ref{ex3}.\\
\begin{example} 
\label{ex3} 
This example is to illustrate Theorem \ref{thm2}. Let $m=n=4$. $R_{i}$ wants $x_{i}$ and knows $x_{i+1}$, where $+$ is modulo-4 addition. Here all possible matrices of the form $(I-B_{SI}A_{SI})$ as given by (\ref{eqn:minnum3}), denoted by $L_{i}$, $i=1,\ldots,16$ are shown below. Only  $L_{5}$ satisfies the requirement given in the proof of Theorem \ref{thm2}. The set of all optimal index codes is given by the collection of all possible basis of the row space of this matrix. They are $28$ in number.  We list out those codes in Table \ref{tabex3}.

\[
L_{1} = \left[\begin{array}{ccccccc}
1 & 0 & 0 & 0 \\
 0& 1 &0 &0\\
 0& 0&1&0\\
 0&0&0&1
\end{array}\right], ~~   L_{2} = \left[\begin{array}{ccccccc}
 1 & 0 & 0 & 0 \\
 0& 1 &0 &0\\
 0& 0&1&0\\
 1&0&0&1
\end{array}\right], 
\]
\[
L_{3} = \left[\begin{array}{ccccccccc}
 1 & 1 & 0 & 0 \\
 0& 1 &0 &0\\
 0& 0&1&0\\
 1&0&0&1
\end{array}\right],   ~~ 
L_{4} = \left[\begin{array}{ccccccccc}
  1 & 1 & 0 & 0 \\
 0& 1 &1 &0\\
 0& 0&1&0\\
 1&0&0&1
\end{array}\right]
\]
\[
L_{5} = \left[\begin{array}{ccccccccc}
  1 & 1 & 0 & 0 \\
 0& 1 &1 &0\\
 0& 0&1&1\\
 1&0&0&1
\end{array}\right],   ~~ 
L_{6} = \left[\begin{array}{ccccccccc}
  1 & 0 & 0 & 0 \\
 0& 1 &1 &0\\
 0& 0&1&0\\
 1&0&0&1
 \end{array}\right],  
\]
\[
L_{7} = \left[\begin{array}{ccccccccc}
  1 & 0 & 0 & 0 \\
 0& 1 &1 &0\\
 0& 0&1&1\\
 1&0&0&1
 \end{array}\right],   ~~ 
L_{8} = \left[\begin{array}{ccccccccc}
  1 & 0 & 0 & 0 \\
 0& 1 &0 &0\\
 0& 0&1&1\\
 1&0&0&1
 \end{array}\right],
\]
\[
L_{9} = \left[\begin{array}{ccccccccc}
  1 & 1 & 0 & 0 \\
 0& 1 &0 &0\\
 0& 0&1&0\\
 0&0&0&1
 \end{array}\right],   ~~ 
L_{10} = \left[\begin{array}{ccccccccc}
  1 & 1 & 0 & 0 \\
 0& 1 &1 &0\\
 0& 0&1&0\\
 0&0&0&1
 \end{array}\right], 
\]
\[
L_{11} = \left[\begin{array}{ccccccccc}
  1 & 1 & 0 & 0 \\
 0& 1 &1 &0\\
 0& 0&1&1\\
 0&0&0&1
 \end{array}\right],   ~~ 
L_{12} = \left[\begin{array}{ccccccccc}
  1 & 1 & 0 & 0 \\
 0& 1 &0 &0\\
 0& 0&1&1\\
 0&0&0&1
 \end{array}\right],
\]
\[
L_{13} = \left[\begin{array}{ccccccccc}
  1 & 0 & 0 & 0 \\
 0& 1 &1 &0\\
 0& 0&1&0\\
 0&0&0&1
 \end{array}\right],   ~~ 
L_{14} = \left[\begin{array}{ccccccccc}
  1 & 0 & 0 & 0 \\
 0& 1 &1 &0\\
 0& 0&1&1\\
 0&0&0&1
 \end{array}\right], 
\]
\[
L_{15} = \left[\begin{array}{ccccccccc}
  1 & 0 & 0 & 0 \\
 0& 1 &0 &0\\
 0& 0&1&1\\
 0&0&0&1
 \end{array}\right],   ~~ 
L_{16} = \left[\begin{array}{ccccccccc}
  1 & 1 & 0 & 0 \\
 0& 1 &0 &0\\
 0& 0&1&1\\
 1&0&0&1
 \end{array}\right]
\]
  \end{example}

\begin{table}
\begin{center}
\small{
\begin{tabular}{|c|c|}
\hline
 {Code} &  {Encoding} \\
 \hline
$\mathfrak{C}_{1}$ & $x_{1}+x_{2},x_{2}+x_{3},x_{3}+x_{4}$   \\
\hline
$\mathfrak{C}_{2}$ & $x_{1}+x_{2},x_{2}+x_{3},x_2+x_4$   \\
\hline
$\mathfrak{C}_{3}$ & $x_{1}+x_{2},x_{2}+x_{3},x_1+x_2+x_3+x_4$  \\
\hline
$\mathfrak{C}_{4}$ & $x_{1}+x_{2},x_{2}+x_{3},x_1+x_4$   \\
\hline
$\mathfrak{C}_{5}$ & $x_{1}+x_{2},x_{3}+x_{4},x_1+x_3$   \\
\hline
$\mathfrak{C}_{6}$ & $x_{1}+x_{2},x_{3}+x_{4},x_2+x_4$   \\
\hline
$\mathfrak{C}_{7}$ & $x_{1}+x_{2},x_{3}+x_{4},x_1+x_4$   \\
\hline
$\mathfrak{C}_{8}$  & $x_{1}+x_{2},x_{1}+x_{3},x_2+x_4$  \\
\hline
$\mathfrak{C}_{9}$ & $x_{1}+x_{2}, x_{1}+x_{3},x_1+x_2+x_3+x_4$   \\
\hline
$\mathfrak{C}_{10}$ & $x_{1}+x_{2}, x_{1}+x_{3},x_1+x_4$  \\
\hline
$\mathfrak{C}_{11}$ & $x_{1}+x_{2},x_{2}+x_{4},x_1+x_2+x_3+x_4$   \\
\hline
$\mathfrak{C}_{12}$ & $x_{1}+x_{2},x_1+x_2+x_3+x_4,x_1+x_4$   \\
\hline
$\mathfrak{C}_{13}$  & $x_{2}+x_{3},x_3+x_4,x_1+x_3$  \\
\hline
$\mathfrak{C}_{14}$ & $x_{2}+x_{3},x_3+x_4,x_1+x_2+x_3+x_4$  \\
\hline
$\mathfrak{C}_{15}$ & $x_{2}+x_{3},x_3+x_4,x_1+x_4$  \\
\hline
$\mathfrak{C}_{16}$ & $x_{2}+x_{3},x_1+x_3,x_2+x_4$  \\
\hline
$\mathfrak{C}_{18}$ & $x_{2}+x_{3},x_1+x_3,x_1+x_4$   \\
\hline
$\mathfrak{C}_{19}$  & $x_{2}+x_{3},x_2+x_4,x_1+x_2+x_3+x_4$   \\
\hline
$\mathfrak{C}_{20}$  & $x_{2}+x_{3},x_2+x_4,x_1+x_4$  \\
\hline
$\mathfrak{C}_{21}$ & $x_{3}+x_{4},x_1+x_3,x_2+x_4$   \\
\hline
$\mathfrak{C}_{22}$  & $x_{3}+x_{4},x_1+x_3,x_1+x_2+x_3+x_4$  \\
\hline
$\mathfrak{C}_{23}$ & $x_{1}+x_{3},x_2+x_4,x_1+x_4$  \\
\hline
$\mathfrak{C}_{24}$ & $x_{1}+x_{3},x_1+x_2+x_3+x_4,x_1+x_4$   \\
\hline
$\mathfrak{C}_{25}$ & $x_{2}+x_{4},x_1+x_2+x_3+x_4,x_1+x_4$   \\
\hline
$\mathfrak{C}_{26}$ & $x_{3}+x_{4},x_2+x_4,x_1+x_2+x_3+x_4$  \\
\hline
$\mathfrak{C}_{27}$  & $x_{3}+x_{4},x_2+x_4,x_1+x_4$   \\
\hline
$\mathfrak{C}_{28}$  & $x_{3}+x_{4},x_1+x_2+x_3+x_4,x_1+x_4$   \\
\hline
\end{tabular}
\caption{\small{ All possible optimal linear solutions for Example \ref{ex3}}.}
\label{tabex3}
}
\end{center}
\end{table}

\begin{corollary} 
\label{cor1}
The number of index codes possible with the optimal length  $c$ for a single unicast IC problem is given by
\begin{equation}
\label{Cor1eq}
\frac{\mu}{c!} \prod_{i=0}^{c-1} (2^{c}-2^{i})
\end{equation}
where $\mu$ is the number of distinct row spaces of $c-$rank RHS matrix of (\ref{eqn:minnum1}) obtainable from all possible  choices  of  $B_{SI}$ matrices  out of the $2^{{\overset{n}{\underset{i=1}{\sum}}} \mid  K_{i}\mid}$ possible ones.
\begin{proof} 
The proof of this follows from that of Theorem \ref{thm2}.
\end{proof}
\end{corollary}
Note that $\mu=1$ for Example \ref{ex1} and Example \ref{ex2}. The following example is a case with $\mu =2.$
\begin{table}
\begin{center}
\scriptsize{
\begin{tabular}{|c|c|c|}
\hline
 {Code} &  {Encoding} & $t_{max}(T)$\\
 \hline
$\mathfrak{C}_{1}$ & $x_{1}+x_{2},x_{2}+x_{3},x_{3}+x_{4}$ &3 \\
\hline
$\mathfrak{C}_{2}$ & $x_{1}+x_{2},x_{2}+x_{3},x_2+x_4$ & 2\\
\hline
$\mathfrak{C}_{3}$ & $x_{1}+x_{2},x_{2}+x_{3},x_1+x_2+x_3+x_4$ &2  \\
\hline
$\mathfrak{C}_{4}$ & $x_{1}+x_{2},x_{2}+x_{3},x_1+x_4$ &3 \\
\hline
$\mathfrak{C}_{5}$ & $x_{1}+x_{2},x_{3}+x_{4},x_1+x_3$ & 2  \\
\hline
$\mathfrak{C}_{6}$ & $x_{1}+x_{2},x_{3}+x_{4},x_2+x_4$ &2 \\
\hline
$\mathfrak{C}_{7}$ & $x_{1}+x_{2},x_{3}+x_{4},x_1+x_4$ &3\\
\hline
$\mathfrak{C}_{8}$  & $x_{1}+x_{2},x_{1}+x_{3},x_2+x_4$ &3 \\
\hline
$\mathfrak{C}_{9}$ & $x_{1}+x_{2}, x_{1}+x_{3},x_1+x_2+x_3+x_4$ &3 \\
\hline
$\mathfrak{C}_{10}$ & $x_{1}+x_{2}, x_{1}+x_{3},x_1+x_4$ &2  \\
\hline
$\mathfrak{C}_{11}$ & $x_{1}+x_{2},x_{2}+x_{4},x_1+x_2+x_3+x_4$ &3 \\
\hline
$\mathfrak{C}_{12}$ & $x_{1}+x_{2},x_1+x_2+x_3+x_4,x_1+x_4$ &2  \\
\hline
$\mathfrak{C}_{13}$  & $x_{2}+x_{3},x_3+x_4,x_1+x_3$ &2 \\
\hline
$\mathfrak{C}_{14}$ & $x_{2}+x_{3},x_3+x_4,x_1+x_2+x_3+x_4$ &2 \\
\hline
$\mathfrak{C}_{15}$ & $x_{2}+x_{3},x_3+x_4,x_1+x_4$ &3 \\
\hline
$\mathfrak{C}_{16}$ & $x_{2}+x_{3},x_1+x_3,x_2+x_4$ &3 \\
\hline
$\mathfrak{C}_{17}$ & $x_{2}+x_{3},x_1+x_3,x_1+x_2+x_3+x_4$ &3 \\
\hline
$\mathfrak{C}_{18}$ & $x_{2}+x_{3},x_1+x_3,x_1+x_4$ &2 \\
\hline
$\mathfrak{C}_{19}$  & $x_{2}+x_{3},x_2+x_4,x_1+x_2+x_3+x_4$ &3 \\
\hline
$\mathfrak{C}_{20}$  & $x_{2}+x_{3},x_2+x_4,x_1+x_4$ &2\\
\hline
$\mathfrak{C}_{21}$ & $x_{3}+x_{4},x_1+x_3,x_2+x_4$ &3 \\
\hline
$\mathfrak{C}_{22}$  & $x_{3}+x_{4},x_1+x_3,x_1+x_2+x_3+x_4$ &3 \\
\hline
$\mathfrak{C}_{23}$ & $x_{1}+x_{3},x_2+x_4,x_1+x_4$ &3 \\
\hline
$\mathfrak{C}_{24}$ & $x_{1}+x_{3},x_1+x_2+x_3+x_4,x_1+x_4$ &3 \\
\hline
$\mathfrak{C}_{25}$ & $x_{2}+x_{4},x_1+x_2+x_3+x_4,x_1+x_4$ &3 \\
\hline
$\mathfrak{C}_{26}$ & $x_{3}+x_{4},x_2+x_4,x_1+x_2+x_3+x_4$ &3 \\
\hline
$\mathfrak{C}_{27}$  & $x_{3}+x_{4},x_2+x_4,x_1+x_4$ &2 \\
\hline
$\mathfrak{C}_{28}$  & $x_{3}+x_{4},x_1+x_2+x_3+x_4,x_1+x_4$ &2 \\
\hline
\end{tabular}
}
\end{center}
\caption{\small Optimal linear solutions corresponding to $B_{SI,1}$ for Example \ref{ex4}.}
\label{tab1}
\end{table}

\begin{table}
\begin{center}
\scriptsize{
\begin{tabular}{|c|c|c|}
\hline
 {Code} &  {Encoding} & $t_{max}(T)$\\
 \hline
$\mathfrak{C}_{29}$  & $x_{3}+x_{2},x_2+x_1,x_1+x_4+x_3$ &3 \\
\hline
$\mathfrak{C}_{30}$  & $x_{3}+x_{2},x_2+x_1,x_4$ &2 \\
\hline
$\mathfrak{C}_{31}$  & $x_{3}+x_{2},x_2+x_1,x_1+x_4+x_2$ &2 \\
\hline
$\mathfrak{C}_{32}$  & $x_{1}+x_{2},x_2+x_3,x_2+x_4+x_3$ &2 \\
\hline
$\mathfrak{C}_{33}$  & $x_{1}+x_{2},x_4,x_1+x_4+x_3$ &3 \\
\hline
$\mathfrak{C}_{34}$  & $x_{1}+x_{2},x_1+x_3,x_1+x_4+x_3$ &2 \\
\hline
$\mathfrak{C}_{35}$  & $x_{1}+x_{2},x_2+x_4+x_1,x_1+x_4+x_3$ &3 \\
\hline
$\mathfrak{C}_{36}$  & $x_{1}+x_{2},x_1+x_3,x_4 $&2 \\
\hline
$\mathfrak{C}_{37}$  & $x_{1}+x_{2},x_4,x_2+x_4+x_3$ &3 \\
\hline
$\mathfrak{C}_{38}$  & $x_{1}+x_{2},x_1+x_3,x_1+x_4+x_2$ &2 \\
\hline
$\mathfrak{C}_{39}$  & $x_{1}+x_{2},x_1+x_3,x_2+x_4+x_3$ &3 \\
\hline
$\mathfrak{C}_{40}$  & $x_{1}+x_{2},x_2+x_3+x_4,x_1+x_4+x_2$ &2 \\
\hline
$\mathfrak{C}_{41}$  & $x_{3}+x_{2},x_4,x_1+x_4+x_3$ &3 \\
\hline
$\mathfrak{C}_{42}$  & $x_{3}+x_{2},x_1+x_3,x_1+x_4+x_3$ &2 \\
\hline
$\mathfrak{C}_{43}$  & $x_{3}+x_{2},x_2+x_3+x_4,x_1+x_4+x_3$ &3 \\
\hline
$\mathfrak{C}_{44}$  & $x_{3}+x_{2},x_1+x_3,x_4$ &2 \\
\hline
$\mathfrak{C}_{45}$  & $x_{3}+x_{2},x_4,x_1+x_4+x_2$ &3 \\
\hline
$\mathfrak{C}_{46}$  & $x_{3}+x_{2},x_1+x_3,x_1+x_4+x_2$ &3 \\
\hline
$\mathfrak{C}_{47}$  & $x_{3}+x_{2},x_1+x_3,x_2+x_4+x_3$ &2 \\
\hline
$\mathfrak{C}_{48}$  & $x_{3}+x_{2},x_2+x_4+x_1,x_2+x_4+x_3$ &3 \\
\hline
$\mathfrak{C}_{49}$  & $x_{1}+x_{2}+x_4,x_4,x_1+x_4+x_3$ &2 \\
\hline
$\mathfrak{C}_{50}$  & $x_{3}+x_{2}+x_4,x_4,x_1+x_4+x_3$ &2 \\
\hline
$\mathfrak{C}_{51}$  & $x_{3}+x_{1},x_2+x_4+x_1,x_1+x_4+x_3$ &3 \\
\hline
$\mathfrak{C}_{52}$  & $x_{3}+x_{2}+x_4,x_2+x_1+x_4,x_1+x_4+x_3$ &3 \\
\hline
$\mathfrak{C}_{53}$  & $x_{3}+x_{1},x_4,x_1+x_4+x_2$ &3 \\
\hline
$\mathfrak{C}_{54}$  & $x_{3}+x_{1},x_4,x_2+x_4+x_3$ &3 \\
\hline
$\mathfrak{C}_{55}$  & $x_4,x_1+x_2+x_4,x_1+x_4+x_3$ &2 \\
\hline
$\mathfrak{C}_{56}$  & $x_{3}+x_{1},x_2+x_3+x_4,x_1+x_4+x_3$ &3 \\

\hline
\end{tabular}
}
\end{center}
\caption{\small Optimal linear solutions corresponding to $B_{SI,2}$ for Example \ref{ex4}.}
\label{tab2}
\end{table}

\begin{example} 
\label{ex4}
This example illustrates Corollary \ref{cor1}. Let $m=n=4$. $R_{i}$ wants $x_{i}$ and knows $x_{i+1}$ where $+$ is modulo-4 operation. $x_{3}$ knows $x_{1}$ also. The optimal length is $c=3$ and it can be checked that  $\mu=2$ and the corresponding $B_{SI}$ matrices are
$B_{SI,1}=
\left[ \begin{array}{ccccc}
1 & 0 & 0 & 0 &0\\
0 & 1& 0 & 0 & 0 \\
0 & 0 & 0 & 1&0\\
0&0&0&0&1
\end{array}
\right]
$ and $B_{SI,2}=
\left[ \begin{array}{ccccc}
1 & 0 & 0 & 0 &0\\
0 & 1& 0 & 0 & 0 \\
0 & 0 & 1 & 1&0\\
0&0&0&0&0
\end{array}
\right]
$. These have been obtained from the general form of $B_{SI}=
\left[ \begin{array}{ccccc}
x & 0 & 0 & 0 &0\\
0 & x& 0 & 0 & 0 \\
0 & 0 & x & x&0\\
0&0&0&0&x
\end{array}
\right]
$ where $x$ can take the values $0$ or $1.$ We have 
$$
A_{SI}=\left[ \begin{array}{cccc}
0 & 1 & 0 & 0 \\
0 & 0& 1 & 0  \\
0 & 0 & 0 & 1\\
1 & 0 & 0 & 0 \\
1 & 0 & 0 & 0
\end{array}
\right]
$$ 
and 
$$
I-B_{SI,1}A_{SI}=\left[ \begin{array}{cccc}
1 & 1 & 0 & 0 \\
0 & 1& 1 & 0  \\
1 & 0 & 1 & 0\\
1 & 0 & 0 & 1 
\end{array}
\right]
$$
whose rank is 3. There are 28 different bases for $I-B_{SI,1}A_{SI}$ one of them being 
$
\left[ \begin{array}{cccc}
1 & 1 & 0 & 0 \\
0 & 1& 1 & 0  \\
0 & 0 & 1 & 1 
\end{array}
\right]
$
which corresponds to the code $\{x_1+x_2,~~ x_2+x_3,~~ x_3+x_4 \}$ consisting of codewords one corresponding to each row of the basis. 

The total number of optimal linear codes are 56 in number, 28 corresponding to each $B_{SI,1}$ and $B_{SI,2}.$ All these 56 codes are listed in Table \ref{tab1} and Table \ref{tab2}. The last column in these tables is described and used in the following section.
\label{mu}
 \end{example}
\begin{corollary}
\label{cor2}
The bound in Theorem 2 is satisfied with equality for  single unicast single uniprior IC problems.
\end{corollary}
\begin{proof}: Consider a single unicast-single uniprior problem with $n$ receivers and $n$ messages. By construction, $B_{SI}=xI_{n\times n}$, which is a diagonal matrix with entries $x$ each one of which can take values from  $\{0,1\}$. Since the $i^{th}$ receiver does not know $x_i$ a priori, $(i,i)^{th}$ element in $A_{SI}$ is $0$ $ \forall i \in \{1,2,\ldots, n\}$. The matrix $A_{SI}$ is a permutation matrix which has no $1$s on the diagonal. It permutes the columns of $B_{SI}$ such that $j^{th}$ column in $B_{SI}A_{SI}$ is not identical to the $j^{th}$ column of $B_{SI}$ for any $j\in\{1,2,\ldots, n\}$. So, the matrix $B_{SI}A_{SI}$ has $0s$ on its diagonal. This implies that the $n \times n$ matrix $I-B_{SI}A_{SI}$ has $1s$ along its diagonal. Its $i^{th}$ row corresponds to the $i^{th}$ receiver $R_i$ and $i^{th}$ column corresponds to the message $x_i$.  In a given row $r_i$, there are at most two $1s$. The $1$ on the diagonal corresponds to the message $x_i$ wanted by $R_i$ and the other 1 (say $(i,j)^{th}$ position) corresponds to the message $x_j$ known a priori by $R_i$.

        Information flow graph $G$ on $n$ nodes can be constructed as follows. By convention, node $i$ corresponds to the receiver that knows message $x_i$. It has one incoming edge $(j,i)$ originating from the node $j$ where $x_j$ is the message wanted by the receiver $R_i$ and known by the receiver $R_j$. It has one outgoing edge $(i,k)$ terminating at node $k$ where $x_i$ is the message wanted by the receiver $R_k$ and known by $R_i$.

        Identify the receiver that knows $x_1$. Suppose $i^{th}$ row has $x$ in the $1^{st}$ column. This means $R_i$ knows $x_1$. Draw arc $(i,1)$. Search along the $1^{st}$ row to identify the message known a priori by $R_1$. Suppose it is $x_j$. Draw arc $(1,j)$. Now, search along the $j^{th}$ row to identify the message known a priori by $R_j$. Since the problem is single unicast, it can be either $x_k$ or $x_1$. If $x_k$, draw arc $(j,k)$ else draw arc $(j,1)$. Each node has only one incoming edge since every receiver wants a unique message. Each node has only one outgoing edge as the message known a priori by the corresponding receiver is demanded by some other node. Since there are a finite number of nodes, $n$ we can conclude that the information-flow graph $G$ for a single unicast-single uniprior problem will be either one cycle of $n$ nodes or a set of disjoint cycles.

    It was shown in \cite{OnH} that for any single-uniprior problem represented by an information-flow graph $G(\mathcal{V},\mathcal{A})$, after executing the Pruning Algorithm, we have
    \begin{equation}
    \label{eqn:SUCSUP}
        l^{*}(G)=\sum_{i=1}^{N_{sub}}(V(G^{'}_{sub,i})-1)+A(G^{'}\setminus G^{'}_{sub})
        \end{equation}
        where $ l^{*}(G)$ is the optimal length of the index code, $A(G^{'} \setminus G^{'}_{sub})$ is the number of arcs in $\mathcal{A}^{'}\setminus \mathcal{A}_{sub}$. $G^{'}=G^{'}_{sub}\cup(G^{'}\setminus G^{'}_{sub})$ where
    $G^{'}_{sub}=\bigcup_{i=1}^{N_{sub}}G^{'}_{sub,i}$ is a graph consisting of non-trivial strongly connected components $\{G^{'}_{sub,i}\}$ and $G^{'} \setminus G^{'}_{sub}\equiv(\mathcal{V}^{'}\setminus\mathcal{V}^{'}_{sub},\mathcal{A}^{'}\setminus\mathcal{A}^{'}_{sub})$.

        When $G$ consists of only cycles, $A(G^{'} \setminus G^{'}_{sub})=0$. Thus, for a Single Unicast-Single Uniprior problem, the optimal length is given by
\begin{equation}
\label{equation:SUCSUP}
 c =\sum_{i=1}^{N_{sub}}(V(G^{'}_{sub,i}-1))
\end{equation}

        Suppose that in the $B_{SI}$ matrix any $1$ is replaced with $0$, i.e., the corresponding side information is not used in decoding. This means a node is removed from the graph $G$. The resulting graph $\mathcal{G}$ will have arcs apart from cycles. Consequently, $A(G^{'} \setminus G^{'}_{sub})$ component will be non-zero. This will increase the value of $l^{*}(G)$ which is the minimum number of transmissions required. Hence, we conclude that every side information bit must be used. Thus there is only one possible choice of $B_{SI} $ matrix in which all $x$ take value $1$. So $\mu=1$.
\end{proof}
        Note that we can easily  compute the optimal length of the single unicast-single uniprior problem without going into the pruning algorithm of \cite{OnH}. This is done by inspecting the `cycles' of $G$ from $I-B_{SI}A_{SI}$ matrix as described in the proof. Thus we have a simpler way of finding the optimal length $c$ for a Single Unicast-Single Uniprior IC problem. 

        There is yet another way of finding the optimal length $c$ for a single unicast-single uniprior problem using appropriate permutations corresponding to $A_{SI}$ of the given problem. We know that $A_{SI}$ is a permutation matrix that permutes the $n$ columns of $B_{SI}$.  Every permutation on a finite set can be written as a cycle or as a product of disjoint cycles. Once we have the cycle decomposition of the permutation corresponding to $A_{SI}$, let $l_1,l_2,\ldots,l_k$ be the lengths of its disjoint cycles. The optimal length is given by 
\begin{equation}
\label{permute}
c=\sum\limits_{i=1}^{k} (l_i-1). 
\end{equation}
This means  that for Single-Unicast-Single-Uniprior problems all the information are available in $A_{SI}.$  These advantages are illustrated in the following two examples.

\begin{example}
\label{ex5}
Consider the Single Unicast-Single Uniprior problem given in Table \ref{tab4} with the number of receivers( equivalently, the number of messages), $n=10$. For this problem,
$$B_{SI}=\left[ \begin{array}{cccccccccc}
x & 0 & 0 & 0 & 0 & 0 & 0 & 0 & 0 & 0\\
0 & x & 0 & 0 & 0 & 0 & 0 & 0 & 0 & 0\\
0 & 0 & x & 0 & 0 & 0 & 0 & 0 & 0 & 0\\
0 & 0 & 0 & x & 0 & 0 & 0 & 0 & 0 & 0\\
0 & 0 & 0 & 0 & x & 0 & 0 & 0 & 0 & 0\\
0 & 0 & 0 & 0 & 0 & x & 0 & 0 & 0 & 0\\
0 & 0 & 0 & 0 & 0 & 0 & x & 0 & 0 & 0\\
0 & 0 & 0 & 0 & 0 & 0 & 0 & x & 0 & 0\\
0 & 0 & 0 & 0 & 0 & 0 & 0 & 0 & x & 0\\
0 & 0 & 0 & 0 & 0 & 0 & 0 & 0 & 0 & x\\
\end{array}
\right]$$
        
Note that $x \in \{0,1\}$. 
$$A_{SI}=\left[ \begin{array}{cccccccccc}
0 & 1 & 0 & 0 & 0 & 0 & 0 & 0 & 0 & 0\\
0 & 0 & 1 & 0 & 0 & 0 & 0 & 0 & 0 & 0\\
0 & 0 & 0 & 1 & 0 & 0 & 0 & 0 & 0 & 0\\
1 & 0 & 0 & 0 & 0 & 0 & 0 & 0 & 0 & 0\\
0 & 0 & 0 & 0 & 0 & 0 & 0 & 1 & 0 & 0\\
0 & 0 & 0 & 0 & 0 & 0 & 1 & 0 & 0 & 0\\
0 & 0 & 0 & 0 & 0 & 1 & 0 & 0 & 0 & 0\\
0 & 0 & 0 & 0 & 0 & 0 & 0 & 0 & 1 & 0\\
0 & 0 & 0 & 0 & 0 & 0 & 0 & 0 & 0 & 1\\
0 & 0 & 0 & 0 & 1 & 0 & 0 & 0 & 0 & 0\\
\end{array}
\right]$$

$$I-B_{SI}A_{SI}=\left[ \begin{array}{cccccccccc}
1 & x & 0 & 0 & 0 & 0 & 0 & 0 & 0 & 0\\
0 & 1 & x & 0 & 0 & 0 & 0 & 0 & 0 & 0\\
0 & 0 & 1 & x & 0 & 0 & 0 & 0 & 0 & 0\\
x & 0 & 0 & 1 & 0 & 0 & 0 & 0 & 0 & 0\\
0 & 0 & 0 & 0 & 1 & 0 & 0 & x & 0 & 0\\
0 & 0 & 0 & 0 & 0 & 1 & x & 0 & 0 & 0\\
0 & 0 & 0 & 0 & 0 & x & 1 & 0 & 0 & 0\\
0 & 0 & 0 & 0 & 0 & 0 & 0 & 1 & x & 0\\
0 & 0 & 0 & 0 & 0 & 0 & 0 & 0 & 1 & x\\
0 & 0 & 0 & 0 & x & 0 & 0 & 0 & 0 & 1\\
\end{array}
\right]$$       
The permutation corresponding to the $A_{SI}$ matrix is 
\[ \left(  \begin{array}{cccccccccc}
1 ~2 ~ 3 ~  4 ~ 5 ~  6 ~ 7 ~ 8 ~ 9  ~10 \\
2 ~3 ~ 4 ~  1 ~ 8 ~ 7  ~ 6 ~ 9 ~  10 ~ 5
\end{array}            \right)  \]

which in terms of cycles is $(1 ~ 2 ~ 3 ~ 4)(6 ~ 7) (5 ~ 8 ~ 9 ~ 10 )$
from which we get the optimal length, using \eqref{permute} to be $7.$ 

     The information-flow graph obtained from the $I-B_{SI}A_{SI}$ matrix for this problem is shown in figure \ref{figure:G1}.

        \begin{figure}[htbp]
                \centering
                \includegraphics[scale=0.75]{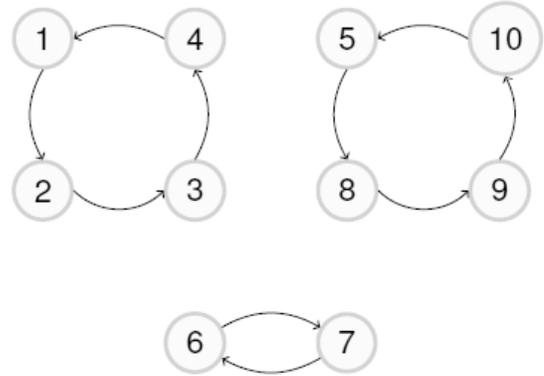}
                \caption{Information flow graph for Example \ref{ex5}}
                \label{figure:G1}
        \end{figure}
Optimal length of the linear index code calculated from the information flow graph is$$c=\sum_{i=1}^{N_{sub}}(V(G^{'}_{sub,i})-1)\\=3+3+1\\=7$$

The number of optimal index codes given by \ref{cor1} is 
$$N_{OIC}=\frac{1}{7!} \prod_{i=0}^{7-1}{(2^{7}-2^{i})}=3.2510\times10^{10}$$

\begin{table}
\label{table3}
\begin{center}
\scriptsize{
\begin{tabular}{|c|c|c|} \hline
{Receiver, $R_i$} &  {Demand set, $\mathcal{W}_i$} & {Side-Information,$\mathcal{K}_i$}\\
\hline  
$R_1$  & $x_1$ &$x_2$ \\
\hline  
$R_2$ & $x_2$ & $x_3$\\ 
\hline  
$R_3$ & $x_3$ & $x_4$\\ 
\hline  
$R_4$ & $x_4$ & $x_1$\\ 
\hline  
$R_5$ & $x_5$ & $x_8$\\ 
\hline  
$R_6$ & $x_6$ & $x_7$\\ 
\hline  
$R_7$ & $x_7$ & $x_6$\\ 
\hline  
$R_8$ & $x_8$ & $x_9$\\ 
\hline  
$R_9$ & $x_9$ & $x_{10}$\\
\hline  
$R_{10}$ & $x_{10}$ & $x_5$\\ 
\hline  
\end{tabular}
}
\end{center}
\caption{\small Single Unicast Single Uniprior problem in Example \ref{ex5}}
\label{tab4}
\end{table}
\end{example}
\begin{example}
\label{ex6}

Consider the Single Unicast-Single Uniprior problem given in Table \ref{tab5} with $n=5$.
For this problem,$$B_{SI}=\left[ \begin{array}{ccccc}
x & 0 & 0 & 0 & 0 \\
0 & x & 0 & 0 & 0 \\
0 & 0 & x & 0 & 0 \\
0 & 0 & 0 & x & 0 \\
0 & 0 & 0 & 0 & x \\
\end{array}
\right]$$        
Note that $x \in \{0,1\}$.
$$A_{SI}=\left[ \begin{array}{ccccc}
0 & 1 & 0 & 0 & 0 \\
0 & 0 & 1 & 0 & 0 \\
1 & 0 & 0 & 0 & 0 \\
0 & 0 & 0 & 0 & 1 \\
0 & 0 & 0 & 1 & 0 \\
\end{array}
\right]$$

$$I-B_{SI}A_{SI}=\left[ \begin{array}{ccccc}
1 & x & 0 & 0 & 0 \\
0 & 1 & x & 0 & 0 \\
x & 0 & 1 & 0 & 0 \\
0 & 0 & 0 & 1 & x \\
0 & 0 & 0 & x & 1 \\
\end{array}
\right]$$        
The permutation corresponding to the $A_{SI}$ matrix is
\[ \left(  \begin{array}{ccccc}
1 ~2 ~ 3 ~  4 ~ 5  \\
2 ~3 ~ 1 ~  5 ~ 4 
\end{array}            \right)  \]

which in terms of cycles is $(1 ~ 2 ~ 3 )(4 ~ 5) $
from which we get the optimal length, using \eqref{permute} to be $3.$

        The information-flow graph obtained from the $I-B_{SI}A_{SI}$ matrix for this problem is shown in figure \ref{figure:G2}.
\begin{figure}[htbp]
        \centering
        \includegraphics[scale=0.75]{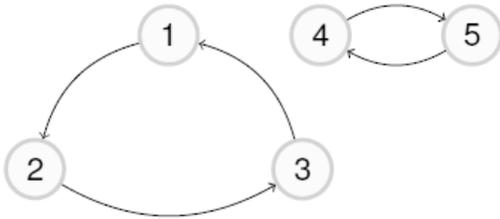}
        \caption{Information-flow graph for Example \ref{ex6}}
        \label{figure:G2}
\end{figure}
        
The optimal length of the linear index code, from the information flow graph is 
$$c=\sum_{i=1}^{N_{sub}}(V(G^{'}_{sub,i})-1)\\=2+1\\=3$$

The number of optimal index codes 
$$N_{OIC}=\frac{1}{3!} \prod_{i=0}^{3-1}{(2^{3}-2^{i})}=28$$

\begin{table}
\label{table4}
\begin{center}
\scriptsize{
\begin{tabular}{|c|c|c|}
\hline  
{Receiver, $R_i$} &  {Demand set, $\mathcal{W}_i$} & {Side-Information,$\mathcal{K}_i$}\\
\hline  
$R_1$  & $x_1$ &$x_2$ \\
\hline  
$R_2$ & $x_2$ & $x_3$\\ 
\hline  
$R_3$ & $x_3$ & $x_1$\\ 
\hline  
$R_4$ & $x_4$ & $x_5$\\ 
\hline  
$R_5$ & $x_5$ & $x_4$\\ 
\hline  

\end{tabular}
}       
\end{center}

\caption{\small Single Unicast Single Uniprior problem in Example \ref{ex6}}
\label{tab5}
\end{table}
        The table \ref{table11} lists these $28$ linear index codes with optimal length, $c=3$. Note that every set of $c$ basis vectors of the row space of $I-B_{SI}A_{SI}$ matrix gives an optimal index code.

        \begin{table}
                \begin{center}
                        \scriptsize{
                                \begin{tabular}{|c|c|}
                                        \hline
                                        {Code, $C_i$} &  {Encoding} \\
                                        \hline
                                        $\mathfrak{C}_1$  & $x_1+x_2$,\quad$x_2+x_3$,\quad$x_4+x_5$\\
                                        \hline
                                        $\mathfrak{C}_2$ & $x_1+x_3$,\quad $x_1+x_2$,\quad $x_4+x_5$\\
                                        \hline
                                        $\mathfrak{C}_{3}$ & $x_1+x_3$,\quad$x_2+x_3$,\quad$x_4+x_5$ \\
                                        \hline
                                        $\mathfrak{C}_4$ & $x_1+x_3$,\quad$x_2+x_3$,\quad$x_1+x_3+x_4+x_5$\\
                                        \hline
                                        $\mathfrak{C}_5$ & $x_1+x_3$,\quad $x_4+x_5$,\quad$x_2+x_3+x_4+x_5$ \\
                                        \hline
                                        $\mathfrak{C}_{6}$ & $x_1+x_3$,\quad$x_2+x_3$,\quad$x_2+x_3+x_4+x_5$ \\
                                        \hline
                                        $\mathfrak{C}_7$ & $x_1+x_2$,\quad $x_1+x_3$,\quad $x_1+x_2+x_4+x_5$\\
                                        \hline
                                        $\mathfrak{C}_8$ & $x_1+x_2$,\quad $x_1+x_3$,\quad$x_1+x_2+x_4+x_5$\\
                                        \hline
                                        $\mathfrak{C}_9$ & $x_1+x_2$,\quad $x_4+x_5$,\quad$x_2+x_3+x_4+x_5$ \\
                                        \hline
                                        $\mathfrak{C}_{10}$&$x_1+x_2$,\quad$x_2+x_3$,\quad$x_1+x_2+x_4+x_5$\\
                                        \hline

                                        $\mathfrak{C}_{11}$ & $x_1+x_2$,\quad$x_2+x_3$,\quad$x_2+x_3+x_4+x_5$ \\
                                        \hline
                                        $\mathfrak{C}_{12}$ & $x_4+x_5$,\quad$x_1+x_3$,\quad$x_1+x_2+x_4+x_5$ \\
                                        \hline
                                        $\mathfrak{C}_{13}$ & $x_4+x_5$,\quad$x_2+x_3$,\quad$x_1+x_2+x_4+x_5$ \\
                                        \hline
                                        $\mathfrak{C}_{14}$ & $x_4+x_5$,\quad$x_1+x_2$,\quad$x_1+x_3+x_4+x_5$ \\
                                        \hline
                                        $\mathfrak{C}_{15}$ &$x_4+x_5$,\quad$x_2+x_3$,\quad$x_1+x_3+x_4+x_5$ \\
                                        \hline
                                        $\mathfrak{C}_{16}$ & $x_1+x_3$,\quad$x_1+x_2+x_4+x_5$,\quad$x_1+x_3+x_4+x_5$ \\
                                        \hline
                                        $\mathfrak{C}_{17}$ & $x_2+x_3$,\quad$x_2+x_3+x_4+x_5$,\quad$x_1+x_3+x_4+x_5$ \\
                                        \hline
                                        $\mathfrak{C}_{18}$ &  $x_1+x_2$,\quad$x_2+x_3+x_4+x_5$,\quad$x_1+x_2+x_4+x_5$ \\
                                        \hline
                                        $\mathfrak{C}_{19}$ & $x_4+x_5$,\quad$x_1+x_3+x_4+x_5$,\quad$x_2+x_3+x_4+x_5$ \\
                                        \hline
                                        $\mathfrak{C}_{20}$ & $x_4+x_5$,\quad$x_1+x_3+x_4+x_5$,\quad$x_1+x_2+x_4+x_5$ \\
                                        \hline
                                        $\mathfrak{C}_{21}$ & $x_4+x_5$,\quad$x_1+x_2+x_4+x_5$,\quad$x_2+x_3+x_4+x_5$ \\
                                        \hline
                                        $\mathfrak{C}_{22}$ & $x_1+x_2$,\quad$x_1+x_2+x_4+x_5$,\quad$x_2+x_3+x_4+x_5$ \\
                                        \hline
                                        $\mathfrak{C}_{23}$ & $x_1+x_2$,\quad$x_1+x_3+x_4+x_5$,\quad$x_1+x_2+x_4+x_5$ \\
                                        \hline
                                        $\mathfrak{C}_{24}$ & $x_1+x_3$,\quad$x_2+x_3+x_4+x_5$,\quad$x_1+x_3+x_4+x_5$ \\
                                        \hline
                                        $\mathfrak{C}_{25}$ & $x_1+x_2$,\quad$x_1+x_3$,\quad$x_2+x_3+x_4+x_5$ \\
                                        \hline
                                        $\mathfrak{C}_{26}$ & $x_2+x_3$,\quad$x_1+x_3$\quad$x_1+x_2+x_4+x_5$ \\
                                        \hline
                                        $\mathfrak{C}_{27}$ & $x_1+x_2$,\quad$x_2+x_3$,\quad$x_1+x_3+x_4+x_5$ \\
                                        \hline
                                        $\mathfrak{C}_{28}$ & $x_1+x_2+x_4+x_5$,$x_1+x_3+x_4+x_5$,$x_2+x_3+x_4+x_5$ \\
                                        \hline
                                \end{tabular}
                        }
                \end{center}

                \caption{\small All possible optimal length index codes for Example \ref{ex6} }
                \label{table11}
        \end{table}
\end{example}
\begin{corollary}
\label{cor3}
For a single unicast-uniprior problem $\mu =1$ and the number of optimal index codes is given by $$ \frac{1}{c!} \prod_{i=0}^{c-1}{(2^{c}-2^{i})}$$
where $c$ is the optimal length. The value of $c$ can be found and it is given by \eqref{eqnSUCUP}.
\end{corollary}
\begin{proof}
Let the number of messages be $n.$ Since the problem  is  single unicast the number of receivers is also  $n.$ Let the receiver $R_i$ want the message $x_i$. In a uniprior case, $\mid \mathcal{K}_i \mid \cap \mid \mathcal{K}_j \mid = \phi $. When a receiver, $R_j$  does not have any side-information, the message demanded by it must be transmitted explicitly. In other words, if there is only one $1$ in row $r_j$(say) of the $I-B_{SI}A_{SI}$ matrix, the message $x_j$ demanded by the corresponding receiver $R_j$ must be transmitted as such. We denote the number of receivers having no side-information to begin with  by $\lambda_{1}$.

Once transmitted, $x_j$ becomes available to all the remaining receivers. Since the problem is unicast, no other receiver wants $x_j$ except for the $j^{th}$ one. Knowledge of $x_j$ to the remaining $n-1$ receivers is thus useless and hence both the receiver $R_j$ as well as the message $x_j$ are removed from further consideration. Note that the receiver that previously had $x_j \in \mathcal{K}_i$ for some $i\neq j$ now has only $\mid \mathcal{K}_i \mid -1$ messages.

We get a new index coding problem with $n-\lambda_1$ receivers and $n-\lambda_1$ messages. This process can be continued till we arrive at a single unicast-single uniprior IC problem. Assuming it took $k$ stages to arrive at a single unicast-single uniprior problem with number of receivers $=$ number of messages $=m$, we can express the minimal length of the original single unicast-uniprior problem as follows:
\begin{equation}
\label{eqnSUCUP}
c=\sum_{i=1}^{N_{sub}}(V(G^{'}_{sub,i}-1))+\sum_{i=1}^{k}\lambda_i
\end{equation}
where the summation term corresponds to the single uniprior single unicast  problem with $m$ messages and receivers.  We know from Corollary \ref{cor3} that $\mu=1$ for Single Unicast-Single Uniprior IC problem. This completes the proof.

\end{proof}

\begin{example}
\label{ex7}
        Consider the Single Unicast-Uniprior problem given in Table \ref{tab7} with number of messages and the number of receivers $n=10$. The optimal length for this problem can be computed in multiple stages as shown in Table \ref{tab7}. 

        \begin{table}

                \begin{center}
                        \footnotesize{
                                \begin{tabular}{|c|c|c|}
                                        \multicolumn{3}{l}{STAGE: 1\qquad\qquad\qquad
                                                $n=10$\qquad\qquad \quad $\lambda_1=3$}\\
                                        \multicolumn{3}{c}{}\\

                                        \hline
                                        {Receiver, $R_i$} &  {Demand set, $\mathcal{W}_i$} & {Side-information, $\mathcal{K}_i$}\\

                                        \hline
                                        $R_1$  & $x_1$ &$x_2$ \\
                                        \hline
                                        $R_2$ & $x_2$ & $x_3$\\
                                        \hline
                                        $R_3$ & $x_3$ & $x_1$\\
                                        \hline
                                        $R_4$ & $x_4$ & $x_5$,$x_7$,$x_8$\\
                                        \hline
                                        $R_5$ & $x_5$ & $x_4$,$x_6$\\
                                        \hline
                                        $R_6$ & $x_6$ & $x_9$\\
                                        \hline
                                        $R_7$ & $x_7$ & $x_{10}$\\
                                        \hline
                                        $R_8$ & $x_8$ & $\phi$\\
                                        \hline
                                        $R_9$ & $x_9$ & $\phi$\\
                                        \hline
                                        $R_{10}$ & $x_{10}$ & $\phi$\\
                                        \hline
                                        \multicolumn{3}{r}{}\\
                                        \multicolumn{3}{l}{STAGE: 2\qquad\qquad\qquad
                                                $n=7$\qquad\qquad\qquad$\lambda_2=2$}\\
                                        \multicolumn{3}{r}{}\\

                                        \hline
                                        {Receiver, $R_i$} &  {Demand set, $\mathcal{W}_i$} & {Side-Information,$\mathcal{K}_i$}\\
                                        \hline
                                        $R_1$ & $x_1$ & $x_2$\\
                                        \hline
                                        $R_2$ & $x_2$ & $x_3$\\
                                        \hline
                                        $R_3$ & $x_3$ & $x_1$\\
                                        \hline
                                        $R_4$ & $x_4$ & $x_5,x_7$\\
                                        \hline
                                        $R_5$ & $x_5$ & $x_4,x_6$\\
                                        \hline
                                        $R_6$ & $x_6$ & $\phi$\\
                                        \hline
                                        $R_7$ & $x_7$ & $\phi$\\

                                        \hline
                                        \multicolumn{3}{r}{SINGLE UNICAST-}\\
                                        \multicolumn{3}{l}{STAGE: 3\qquad\qquad
                                                $n=5$\qquad\qquad\qquad  SINGLE UNIPRIOR}\\
                                        \multicolumn{3}{r}{$c=3$\qquad}\\

                                        \hline
                                        {Receiver, $R_i$} &  {Demand set, $\mathcal{W}_i$} & {Side-Information,$\mathcal{K}_i$}\\
                                        \hline
                                $R_1$ & $x_1$ & $x_2$\\
                                \hline
                                $R_2$ & $x_2$ & $x_3$\\
                                \hline
                                $R_3$ & $x_3$ & $x_1$\\
                                \hline
                                $R_4$ & $x_4$ & $x_5$\\
                                \hline
                                $R_5$ & $x_5$ & $x_4$\\
                                \hline

                                \end{tabular}
                        }
                \end{center}
                \caption{\small Single Unicast-Uniprior problem in Example \ref{ex7}}   \label{tab7}
        \end{table}

From Example \ref{ex2}, the minimal length of the single unicast single uniprior problem in STAGE 3 is 3. The optimal length of the Single Unicast-Uniprior problem,$$c=3+\sum_{i=1}^{k=2}\lambda_i=3+3+2=8$$.

The number of optimal index codes,$$N_{OIC}= \frac{1}{c!} \prod_{i=0}^{c-1}{(2^{c}-2^{i})}=1.3264\times10^{14}$$
\end{example}

In some cases, we may not end up in a Single Unicast-Single Uniprior problem at all and the optimal number of transmissions will be same as the number of messages.The following illustrates such an example.
\begin{example}
\label{ex8}
        Consider the Single Unicast-Uniprior problem given Table \ref{tab8}.

\begin{table}[h!]

        \begin{center}
                \footnotesize{
                        \begin{tabular}{|c|c|c|}
                                        \multicolumn{3}{l}{STAGE: 1\qquad\qquad\qquad
                                                $m=n=10$\qquad\qquad\qquad$\lambda_1=4$}\\
                                        \multicolumn{3}{r}{}\\

                                        \hline
                                {Receiver, $R_i$} &  {Demand set, $\mathcal{W}_i$} & {Side-Information,$\mathcal{K}_i$}\\
                                \hline
                                $R_1$  & $x_1$ &$x_2,x_3,x_4$ \\
                                \hline
                                $R_2$ & $x_2$ & $x_9,x_{10}$\\
                                \hline
                                $R_3$ & $x_3$ & $x_5$\\
                                \hline
                                $R_4$ & $x_4$ & $x_6$\\
                                \hline
                                $R_5$ & $x_5$ & $x_7$\\
                                \hline
                                $R_6$ & $x_6$ & $x_8$\\
                                \hline
                                $R_7$ & $x_7$ & $\phi$\\
                                \hline
                                $R_8$ & $x_8$ & $\phi$\\
                                \hline
                                $R_9$ & $x_9$ & $\phi$\\
                                \hline
                                $R_{10}$ & $x_{10}$ & $\phi$\\
                                \hline
                            \multicolumn{3}{r}{}\\
                                \multicolumn{3}{l}{STAGE: 2\qquad\qquad\qquad
                                                                        $m=n=6$\qquad\qquad\qquad$\lambda_2=3$}\\
                                \multicolumn{3}{r}{}\\
                                \hline
                                {Receiver, $R_i$} &  {Demand set, $\mathcal{W}_i$} & {Side-Information,$\mathcal{K}_i$}\\
                                \hline
                            $R_1$ & $x_1$ & $x_2$,$x_3$,$x_4$\\
                            \hline
                                $R_2$ & $x_2$ & $\phi$\\
                                \hline
                                $R_3$ & $x_3$ & $x_5$\\
                                \hline
                                $R_4$ & $x_4$ & $x_6$\\
                                \hline
                                $R_5$ & $x_5$ & $\phi$\\
                                \hline
                                $R_6$ & $x_6$ & $\phi$\\

                                \hline
                                        \multicolumn{3}{r}{}\\
                                        \multicolumn{3}{l}{STAGE: 3\qquad\qquad\qquad
                                                               $m=n=3$\qquad\qquad\qquad$\lambda_3=2$}\\
                                    \multicolumn{3}{r}{}\\

                                \hline
                                        {Receiver, $R_i$} &  {Demand set, $\mathcal{W}_i$} & {Side-Information,$\mathcal{K}_i$}\\
                                        \hline
                                $R_1$ & $x_1$ & $x_3$,$x_4$\\
                                \hline
                                $R_3$ & $x_3$ & $\phi$\\
                                \hline
                                $R_4$ & $x_4$ & $\phi$\\
                                        \hline
                                        \multicolumn{3}{r}{}\\
                                        \multicolumn{3}{l}{STAGE: 4\qquad\qquad\qquad
                                                $n=1$\qquad\qquad\qquad$\lambda_4=1$}\\
                                        \multicolumn{3}{r}{}\\

                                \hline
                                $R_1$&$x_1$&$\phi$\\
                                \hline
                        \end{tabular}
                }
        \end{center}

        \caption{\small Single Unicast-Uniprior problem in Example \ref{ex8}}   \label{tab8}
\end{table}

 The receivers $R_7$,$R_8$,$R_9$,$R_{10}$ do not have any side-information and hence $x_7$,$x_8$,$x_9$ and $x_{10}$ are transmitted. So, $\lambda_1=4$.
Now the problem reduces to $n=5$ receivers and is represented in STAGE:2 of Table \ref{tab8}.

 Since the receivers demanding $x_2,x_5$ and $x_6$ have no side-information, their demands have to be transmitted explicitly,i.e., $\lambda_2=3$. Thus, total number of transmissions, $$c=4+3+2+1=10.$$
        Note that, in this example, the $\sum_{i=1}^{N_{sub}}(V(G^{'}_{sub,i}-1))$ term in  (\ref{eqnSUCUP}) is $0$.

\end{example}


\begin{corollary}
\label{cor4}
For a single uniprior-unicast problem, $\mu=1$ and the number of optimal index codes is given by $$ \frac{1}{c!} \prod_{i=0}^{c-1}{(2^{c}-2^{i})}$$
where $c$ is the optimal length.
\end{corollary}

\begin{proof}
Let the number of messages be $n$ and the number of receivers be $m.$ Let receiver $R_i$ know a single unique message $x_i$ a priori. Clearly $m \leq n.$ If $m < n,$ then there are $n-m$ messages which are not known to any of the receivers apriori and hence these messages need to be transmitted individually. We assume that this is done and after this, the remaining problem reduces to that of the case where $m=n=n_1.$  Now, let the receiver $R_i$ want $\mathcal{W}_i$ messages. Since the problem is unicast, unless each  $\mathcal{W}_i$ has only one element, there will be receivers having apriori information but not wanting any message. Such receivers can be removed from further consideration. After removing such receivers  let the number of receivers in the  problem be $m_1.$   Now if $m_1=n_1$ then we have a problem of single uniprior single unicast for which $\mu=1.$  On the contrary, if $m_1 < n_1,$ then we repeat the process and eventually we will end up with a single uniprior single unicast problem and will have $\mu=1$  by Corollary \ref{cor2}. This completes the proof.
\end{proof}
\begin{example}
\label{ex9}
Consider the Single Uniprior-Unicast problem given in Table \ref{tab9} with number of messages, $n=10$ and number of receivers, $m=8$.

        The messages $x_{10}$ and $x_9$ have to be transmitted explicitly since they are not part of any receiver's side-information. Once done, the receiver $R_7$ can be eliminated from further consideration as its demand has been met. The problem reduces to STAGE:2 with $m=7$ receivers and $n=8$ messages. Again, since $x_7$ is not part of any side-information it has to be transmitted explicitly. The problem now reduces to Single Unicast-Single Uniprior case with $n=m=7$.
        The optimal length of this IC is $$c=2+1+4=7$$
        The number of optimal index codes,$$N_{OIC}= \frac{1}{7!} \prod_{i=0}^{6}{(2^{7}-2^{i})}=3.251\times 10^{10}$$
        \begin{table}

                \begin{center}
                        \footnotesize{
                                \begin{tabular}{|c|c|c|}
                                        \multicolumn{3}{l}{STAGE: 1\qquad\qquad\qquad
                                                $m=8,n=10$\qquad\qquad \quad $\lambda_1=2$}\\
                                        \multicolumn{3}{c}{}\\

                                        \hline
                                        {Receiver, $R_i$} &  {Demand set, $\mathcal{W}_i$} & {Side-information, $\mathcal{K}_i$}\\

                                        \hline
                                        $R_1$  & $x_3$ &$x_1$ \\
                                        \hline
                                        $R_2$ & $x_1$ & $x_2$\\
                                        \hline
                                        $R_3$ & $x_2$ & $x_3$\\
                                        \hline
                                        $R_4$ & $x_5$,$x_{10}$ & $x_4$\\
                                        \hline
                                        $R_5$ & $x_4$ & $x_5$\\
                                        \hline
                                        $R_6$ & $x_8$ & $x_6$\\
                                        \hline
                                        $R_7$ & $x_9$ & $x_7$\\
                                        \hline
                                        $R_8$ & $x_7$,$x_6$ & $x_8$\\
                                        \hline
                                        \multicolumn{3}{r}{}\\
                                        \multicolumn{3}{l}{STAGE: 2\qquad\qquad\qquad
                                                $m=7,n=8$\qquad\qquad\qquad$\lambda_2=1$}\\
                                        \multicolumn{3}{r}{}\\

                                        \hline
                                        {Receiver, $R_i$} &  {Demand set, $\mathcal{W}_i$} & {Side-Information,$\mathcal{K}_i$}\\

                                        \hline
                                        $R_1$  & $x_3$ &$x_1$ \\
                                        \hline
                                        $R_2$ & $x_1$ & $x_2$\\
                                        \hline
                                        $R_3$ & $x_2$ & $x_3$\\
                                        \hline
                                        $R_4$ & $x_5$ & $x_4$\\
                                        \hline
                                        $R_5$ & $x_4$ & $x_5$\\
                                        \hline
                                        $R_6$ & $x_8$ & $x_6$\\
                                        \hline
                                        $R_8$ & $x_7$,$x_6$ & $x_8$\\

                                        \hline
                                        \multicolumn{3}{r}{SINGLE UNICAST-}\\
                                        \multicolumn{3}{l}{STAGE: 3\qquad\qquad
                                                $n=m=7$\qquad\qquad\qquad  SINGLE UNIPRIOR}\\
                                        \multicolumn{3}{r}{$c=4$\qquad}\\

                                        \hline
                                        {Receiver, $R_i$} &  {Demand set, $\mathcal{W}_i$} & {Side-Information,$\mathcal{K}_i$}\\
                                                \hline
                                                $R_1$  & $x_3$ &$x_1$ \\
                                                \hline
                                                $R_2$ & $x_1$ & $x_2$\\
                                                \hline
                                                $R_3$ & $x_2$ & $x_3$\\
                                                \hline
                                                $R_4$ & $x_5$ & $x_4$\\
                                                \hline
                                                $R_5$ & $x_4$ & $x_5$\\
                                                \hline
                                                $R_6$ & $x_8$ & $x_6$\\
                                                \hline
                                                $R_8$ & $x_6$ & $x_8$\\
                                                \hline

                                \end{tabular}
                        }
                \end{center}

                \caption{\small Single Uniprior-Unicast problem in Example \ref{ex9}}   \label{tab9}
        \end{table}
\end{example}

\section{Optimal Codes with Minimum-Maximum Error Probability}
\label{sec6}
There can be several linear optimal solutions in terms of least bandwidth for an IC problem but among them we try to identify the index code which minimizes the maximum number of transmissions that is required by any receiver in decoding its desired message. The motivation for this is that each of the transmitted symbols is error prone and the lesser the number of transmissions used for decoding the desired message, lesser will be its probability of error. Hence among all the codes with the same length of transmission, the one for which the maximum number of transmissions used by any receiver is the minimum, will have minimum-maximum error probability amongst all the receivers. We give a method to find the best linear solution in terms of minimum-maximum error probability among all the receivers and among all codes with the optimal length $c_{opt}$. For simplicity, throughout the rest of this section, the length  $c$ will mean the optimal length. Each  $T\in S(c)$ corresponds to a unique pair of encoding-decoding operations. For the same index code there can be more than one way of decoding at each receiver. For each set of decoding operations at the receivers, $T$ matrix differs. For a $T\in S(c)$, the corresponding matrix $T_{BC}$ has $n$ rows and $nc$ columns. Let $t_i, ~~ i=1,2,\cdots n,$ denote the $i-$th row of $T_{BC}.$ Denoting this $i-$th row as [$r_{i,1} r_{i,2}.....r_{i,nc}$], we define $t_{i,use}$ for this row as 
\begin{equation}
\label{tiuse}
t_{i, use}=\sum_{j=1}^{c} 1- (I_{{r_{i,j}=0}}  I_{{r_{i,j+c}}=0}...I_{{r_{i,j+(n-1)c}}=0})
\end{equation}
\noindent where $I_{z}$ is the indicator function which is one when event $z$ occurs.  Note that $t_{i,use}$ is the number of transmissions that are used by the $i-$th receiver $R_i$ out of $c$ transmissions.  Also define 
\begin{equation}
\label{tiusemax}
 t_{max}(T) = \underset{i} \max ~~ t_{i,use} 
\end{equation}
and 
\begin{equation}
\label{tiuseminmax}
T_{minmax} = arg \left\{ \underset{T \in S(c)} \min ~~ t_{max}(T) \right\} 
\end{equation}

Note that $T_{minmax}$ is not unique; there may be several such matrices in $S(c).$
\begin{theorem} For the optimal length $c$, any matrix $T_{minmax}$ in $S(c)$ gives an IC with the minimum-maximum error probability among all the receivers. Also, the matrix formed by taking every $n$-th  row of the corresponding $ F_{BC} $ matrix is an optimal linear solution in terms of minimum-maximum error probability  using $c$ number of transmissions.
\begin{proof}
 For the fixed optimal length $c$, $B_{BC}$ matrix will be as given in \eqref{eqn:Bsub}. For any $T\in S(c)$, the number of transmissions used by the $i$-th receiver is given by the number of non-zero entries in $i$-th row of $ B_{BC} $. 
When for example $t$-th $ \epsilon$ element in the $i$-th row of $ B_{BC} $ is zero, the $i$-th element of every $(t+(k-1)c)$-th column for $k=1$ to $ n$, in $ T_{BC} $  turns 0. Hence the number of transmissions used by it is proportional to the $t_{i,use}.$ Therefore, our claim is proved. Moreover the corresponding $F_{BC}$ is the matrix which decides the message flowing in the broadcast channels. So the matrix formed by taking every $n$-th row of $F_{BC}$ is the corresponding Index code .
\end{proof}
\end{theorem}
From the theorem above it is clear that to minimize the maximum probability of the receivers one needs to pick that $T \in S(c)$ for which the the maximum number of nonzeros in a row in the corresponding $B_{BC}$ matrix is minimized.   
\begin{example}
\label{ex10}
 Let $ m=n=3$. Each $R_{i}$ wants $x_{i}$ and knows $x_{i+1}$, where $+$ is mod-3 addition. The optimal length of a linear IC solution for this problem is $2$
  For Example 1, we found out the optimal IC's : They are 1. $\mathfrak{C}_{1}: x_{1}   \oplus   x_{2} $, $ x_{2}   \oplus   x_{3} $ 2. $\mathfrak{C}_{2}: x_{1}   \oplus   x_{3} $, $ x_{3}   \oplus   x_{2} $ 3. $\mathfrak{C}_{3}: x_{1}  \oplus   x_{3} $, $ x_{1}   \oplus  x_{2} $. For all the three, the maximum number of transmissions used by any receiver is $2$. This is verified below. We find the $T_{B}$ matrices for each case as follows:\\
 \[T_{B,1} =\left[ \begin{array}{cccccc}
 1 & 0 & 1 & 0 & 0 &0\\
 0 & 0 & 0 & 1 & 0 &1\\
 1 &0 & 1& 1& 0& 1
 \end{array} \right]\]\\
  \[T_{B,2} =\left[ \begin{array}{cccccc}
 1 & 0 & 0 & 1 & 1 &1\\
 0 & 0 & 0 & 1 & 0 &1\\
 1 &0 & 0& 0& 1& 0
\end{array} \right]\]\\
\[T_{B,3} =\left[ \begin{array}{cccccc}
 0 & 1 & 0 & 0 & 0 &1\\
 1 & 1 & 1 & 0 & 0 &1\\
 1 &0 & 1& 0& 0& 0
 \end{array} \right]\]\\
 It can be verified that $t_{max}(T)$ for all the three is 2. Hence 2 transmissions at most are used by any receiver to decode in all the three cases. The BEP (Bit Error Probability) versus SNR curves for each of the three codes at various receivers are given in Fig.\ref{Code13}, Fig.   \ref{Code23}, Fig. \ref{Code33}. We considered BPSK modulation in Rayleigh faded channel. In Fig.\ref{max3}, the worst case BEP curves for each of the three codes are plotted. We can see that the curves lie on top of each other which proves our claim that all the three codes are equally good in terms of minimum-maximum error probability.\\
\end{example}
\begin{example}
\label{ex11}
Let $m=n=4$. $R_{i}$ wants $x_{i}$ and knows $x_{i+1}$ where $+$ is modulo-4 operation. $R_{3}$ knows $x_{1}$ also.

The optimal length  $c=3$. Tables \ref{tab1} and \ref{tab2} gives the $t_{max}(T)$ for each of the optimal linear codes. The minimum $r_{min}(T)$ is 2. The BER versus SNR curve for $\mathfrak{C}_{30}$ whose $r_{max}(T)=2$ is as in Fig .\ref{30}.  The BER versus SNR curve for $\mathfrak{C}_{29}$ whose $t_{max}(T)=3$ is as in Fig .\ref{29}. The worst case Performance of both codes are plotted in Fig. \ref{worstmu}.

We can observe that the worst performance of $\mathfrak{C}_{30}$ is better than worst performance of $\mathfrak{C}_{29}$.
\end{example}

\begin{figure}[htbp]
\centering
\includegraphics[height=8 cm, width=9.75 cm]{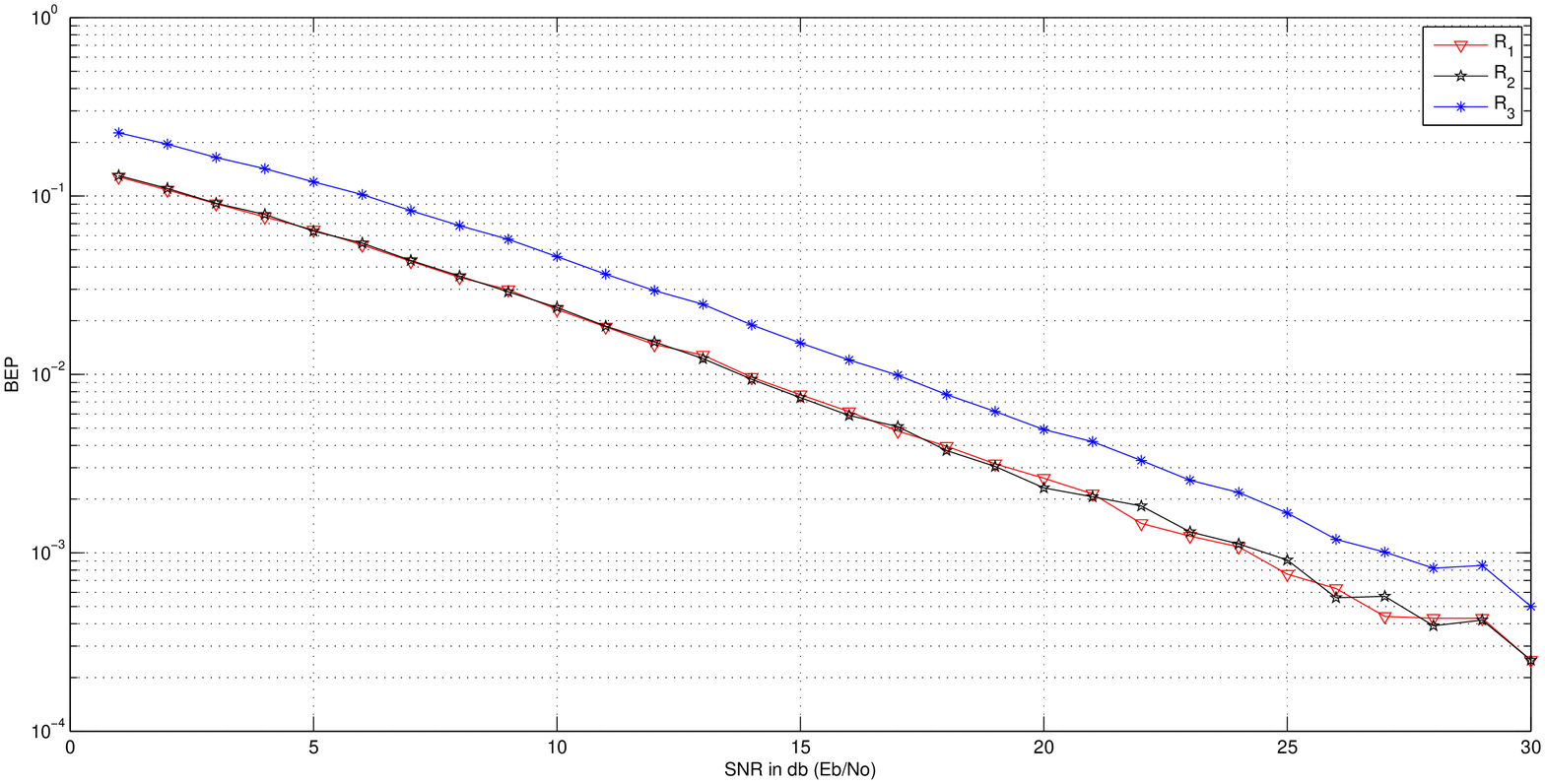}
\caption{BER versus SNR (db) curve for $\mathfrak{C}_{1}$ at all the receivers for Example 1. }
\label{Code13}
\end{figure}
 \begin{figure}[htbp]
\centering
\includegraphics[height=8 cm, width=9.75 cm]{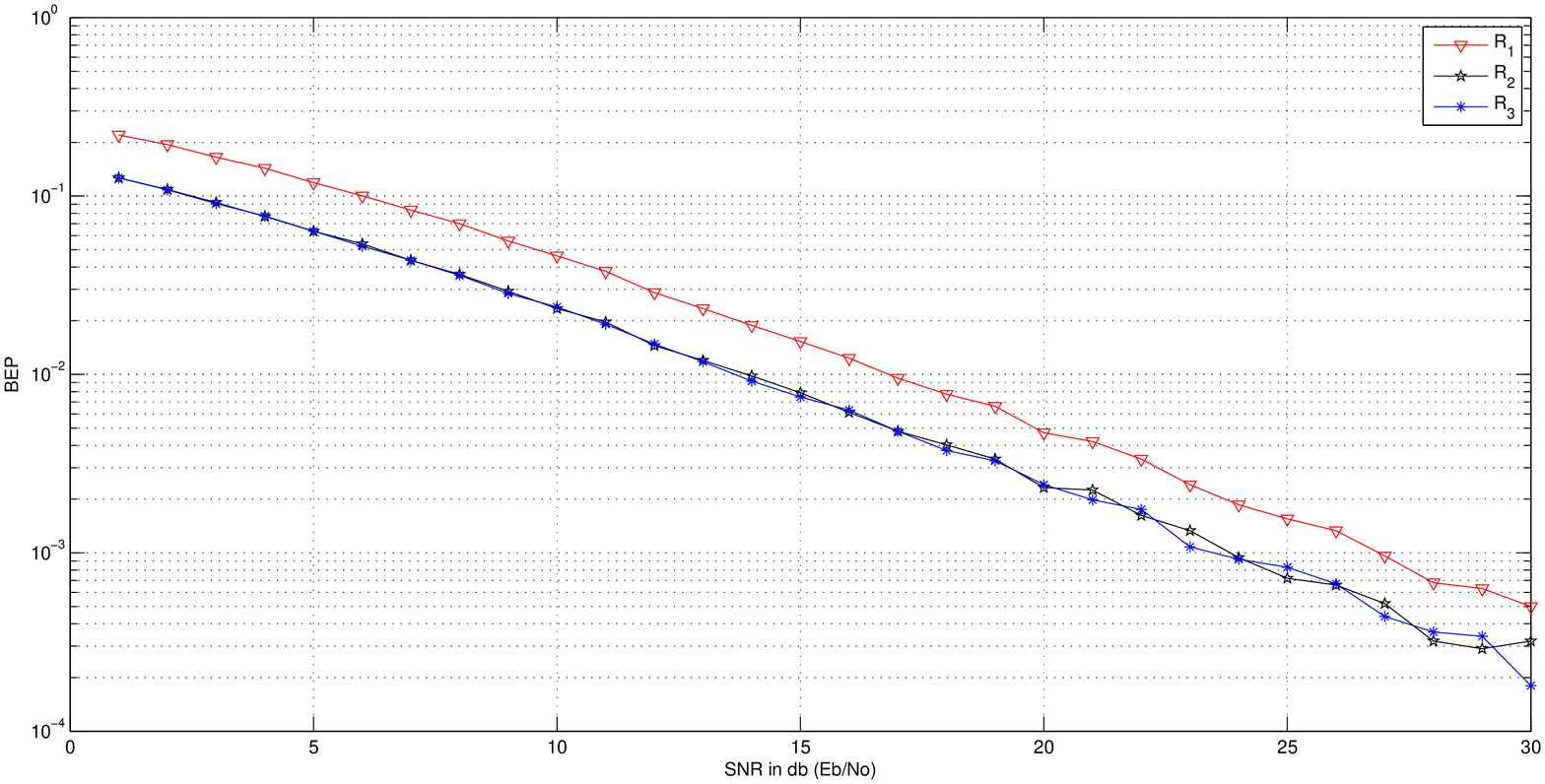}
\caption{BER versus SNR (db) curve for $\mathfrak{C}_{2}$ at all the receivers for Example 1. }
\label{Code23}
\end{figure}
 \begin{figure}[htbp]
\centering
\includegraphics[height=8 cm, width=9.75 cm]{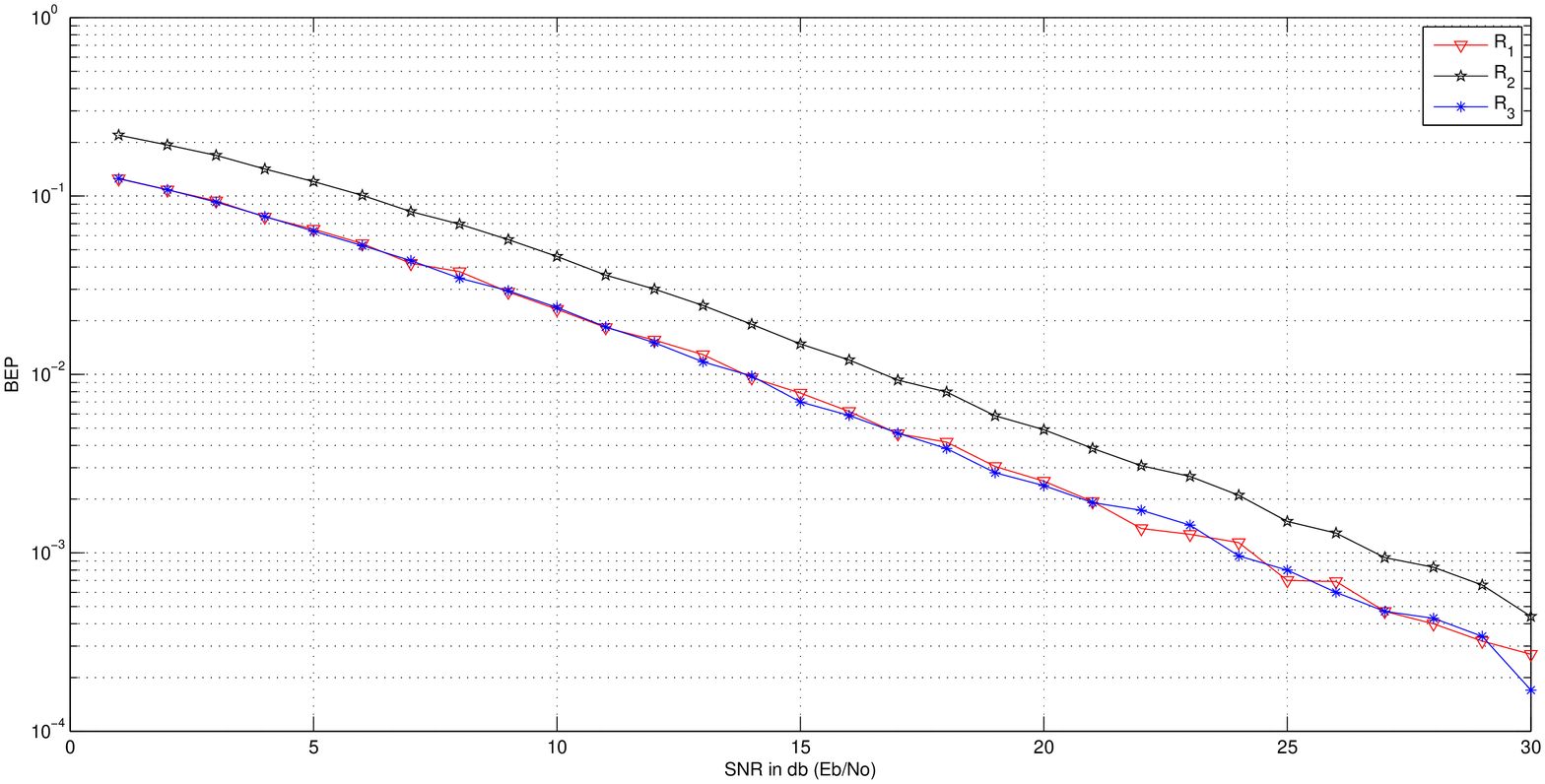}
\caption{BER versus SNR (db) curve for $\mathfrak{C}_{3}$ at all the receivers for Example 1. }
\label{Code33}
\end{figure}
\begin{figure}[htbp]
\centering
\includegraphics[height=8 cm, width=9.75 cm]{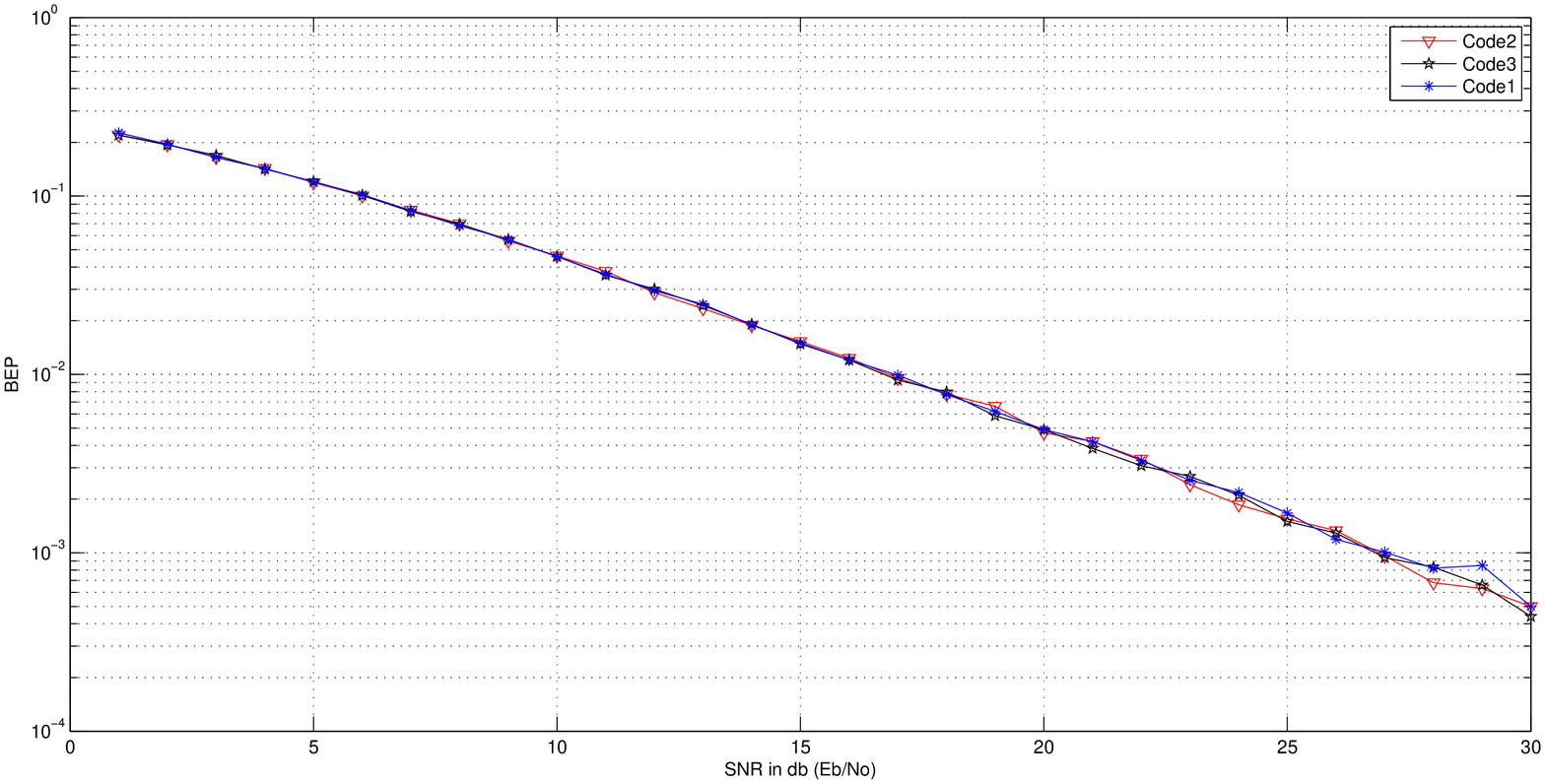}
\caption{ Worst case BER versus SNR (db) curves for each of the three codes for Example 1. }
\label{max3}
\end{figure}
\begin{figure}[htbp]
\centering
\includegraphics[height=8 cm, width=9.75 cm]{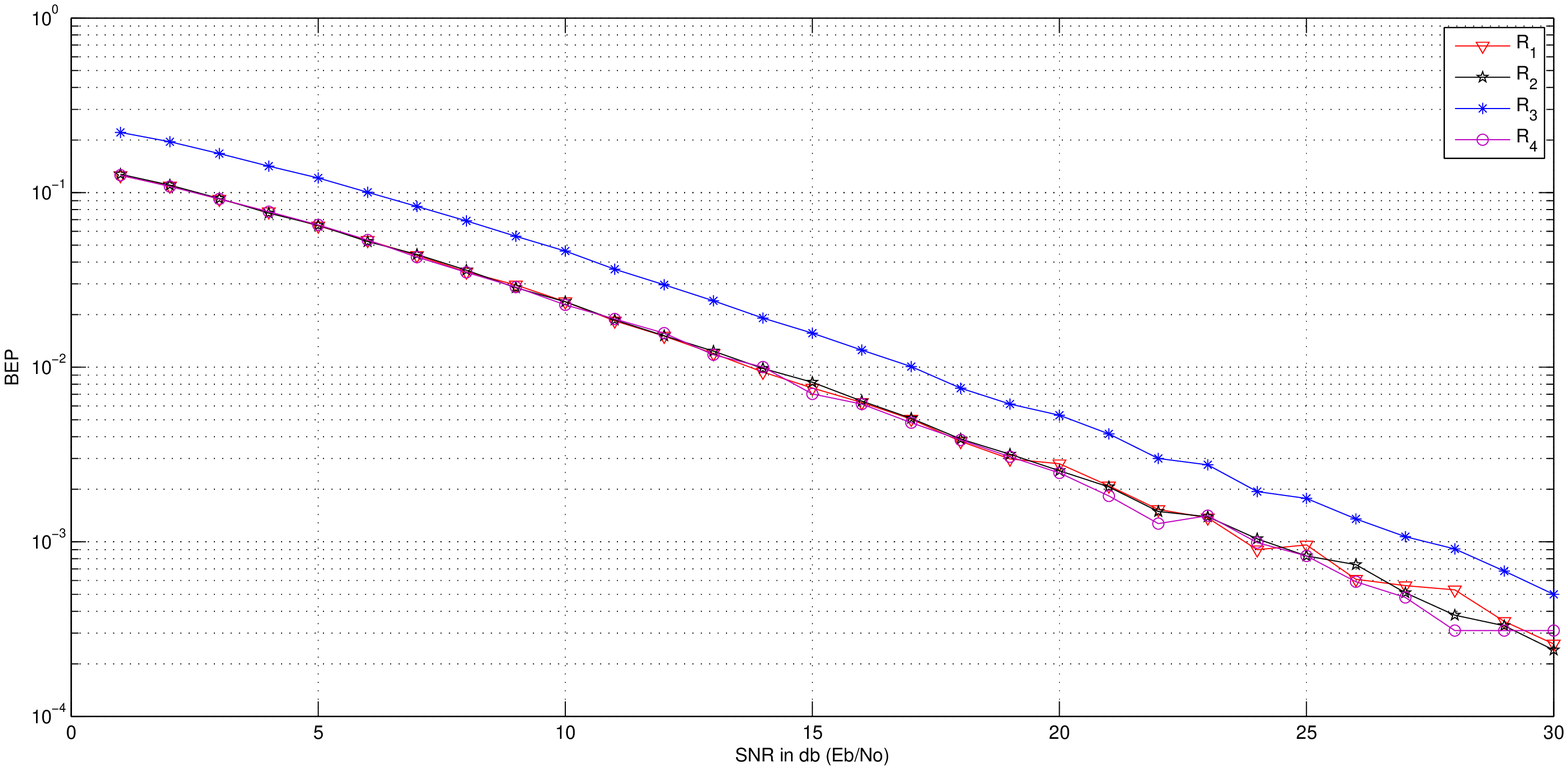}
\caption{BER versus SNR (db) curve for $\mathfrak{C}_{30}$ at all the receivers for Example \ref{mu}}
\label{30}
\end{figure}
\begin{figure}[htbp]
\centering
\includegraphics[height=8 cm, width=9.75 cm]{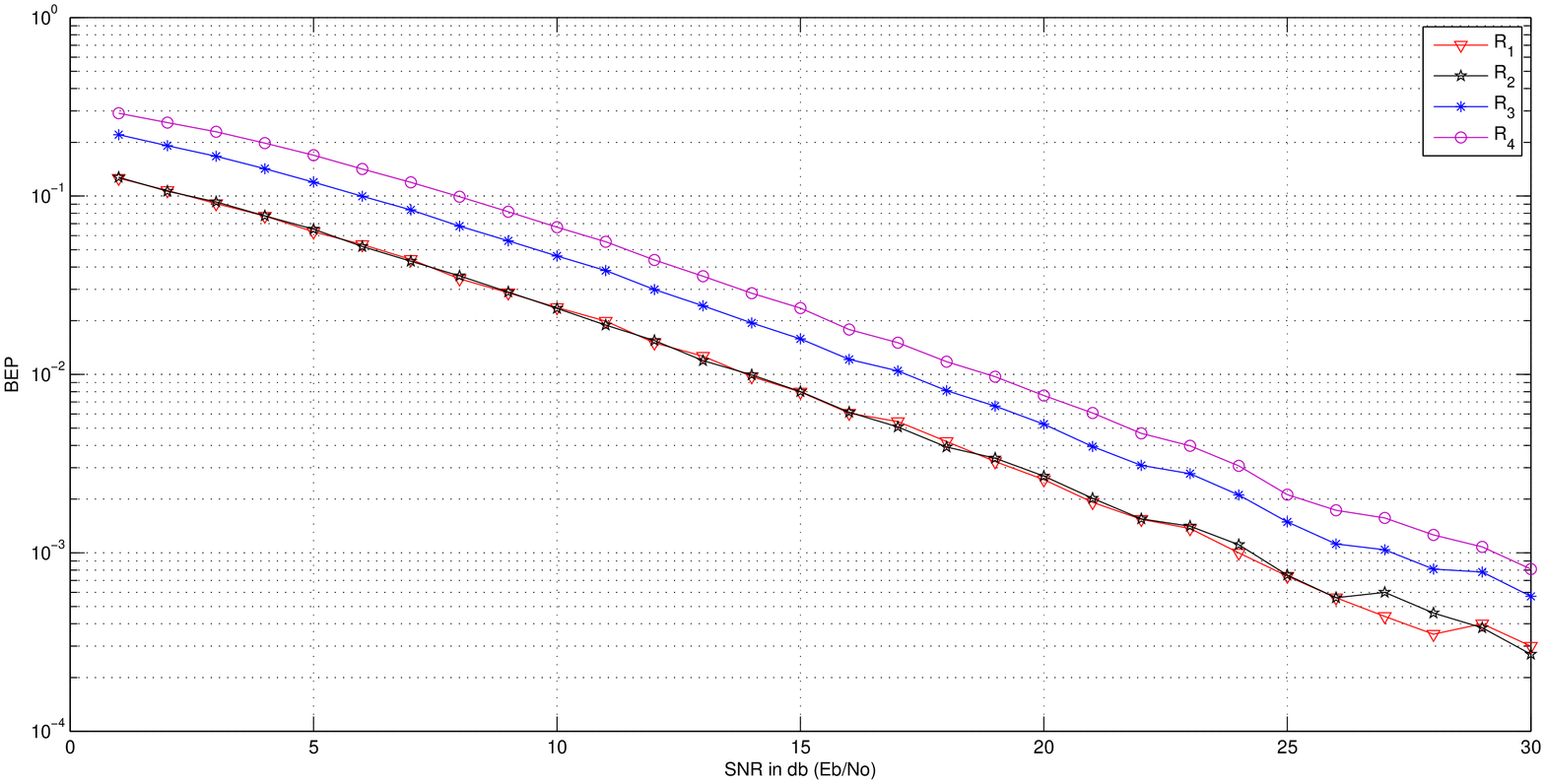}
\caption{BER versus SNR (db) curve for $\mathfrak{C}_{29}$ at all the receivers for Example \ref{mu}}
\label{29}
\end{figure}
\begin{figure}[htbp]
\centering
\includegraphics[height=8 cm, width=9.75 cm]{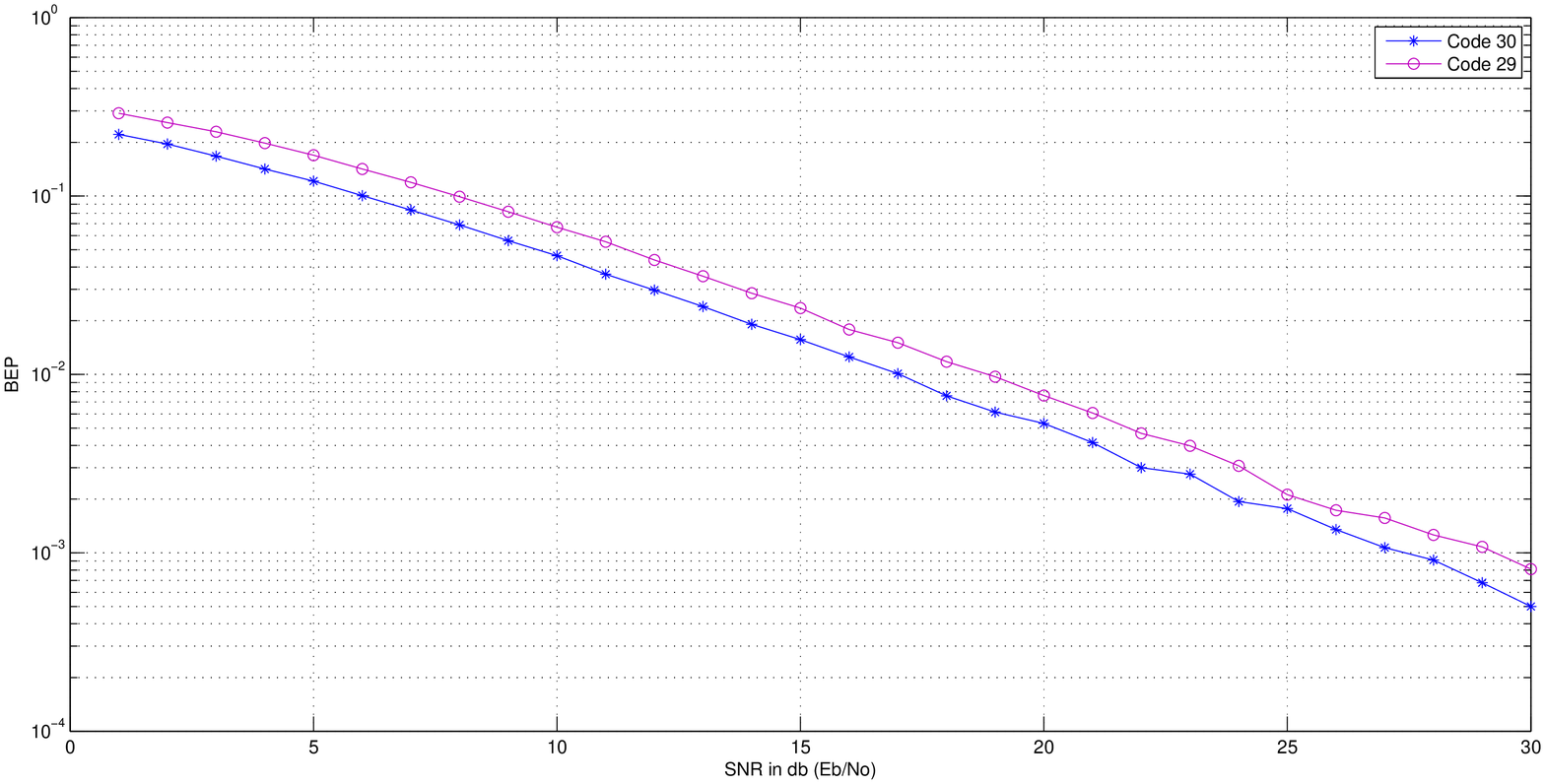}
\caption{Worst case BER versus SNR (db) curves for codes $\mathfrak{C}_{30}$ and $\mathfrak{C}_{29}$   for Example \ref{mu}}
\label{worstmu}
\end{figure}

 \section{Discussion}
\label{sec7}
  For an optimal $c$, the maximum number of index codes possible is bounded by $2^{nc}$. In this paper we have given a lower bound also. This lower bound is satisfied with equality for a single unicast problem in which $\mid K_{i}\mid=1 $ and $K_{i} \cap K_{j}$=0, for $i \neq j, i,j$ = 1 to n . We would like to extend this work to find out least complexity algorithms which finds IC solutions by matrix completion. Harvey \textit{et. all} in [9] gives such algorithms for  multicast network codes. However what we have is a general problem and hence their results are not applicable. We have followed an approach which is different and simpler than Koetter and Medard's [5] for this specific class of network coding problem.
 \section*{Acknowledgement}
 The authors would like to thank Anoop Thomas for the useful discussions on the topic.



\begin{thebibliography}{160}
\bibitem{OnH}
L Ong and C K Ho,``Optimal Index Codes for a Class of Multicast Networks with Receiver Side Information,'' in \textit{Proc. IEEE ICC}, 2012, pp 2213-2218.
\bibitem{YBJK}
Z. Bar-Yoseef, Y. Birk, T. S. Jayram and T. Kol, ``Index coding with side information", in \textit{ IEEE Trans. Inf. Theory}, vol.57, no.3, pp.1479--1494, Mar. 2011.
\bibitem{BiK}
Y. Birk and T. Kol, ``Coding on demand by an informed source (ISCOD) for efficient broadcast of different supplemental data to caching clients", in \textit{IEEE Trans. Inf. Theory}, vol.52, no.6, pp. 2825--2830, Jun. 2006.
\bibitem{RSG}
S.E. Rouayheb, A. Sprintson and C. Georghiades, ``On the Index Coding Problem and its relation to Network Coding and Matroid Theory", in \textit{IEEE Trans. Inf. Theory}, vol.56, no.7, pp. 3187--3195, Jul. 2011.

\bibitem{KaR}
Kavitha Radhakumar and B. Sundar Rajan, ``On the number of optimal index codes,'' Proceedings of IEEE International Symposium on Information Theory, (ISIT 2015), Hong Kong, 14-19 June 2015, pp.1044-1048.


\bibitem{KoM}
Ralf Koetter and Muriel Medard,``An Algebraic Approach to Network Coding", in \textit{IEEE/ACM transactions on networking}, vol.11, no.5, pp. 782--795, Oct. 2003.
\bibitem{Str}
Gilbert Strang, Introduction to Linear Algebra, 3rd ed. MA:Wellesley Cambridge, 2003
\bibitem{TKCR}
Anoop Thomas, Kavitha R., A. Chandramouli, and B. Sundar Rajan, ``Optimal Index Coding with Min-Max
Probability of Error over Fading Channels", Proceedings of IEEE International Symposium on Personal, Indoor and Mobile Radio Communications, (PIMRC 2015), Hong Kong, August 30- Sept' 2, 2015. Also available on ArXiv at http://arxiv.org/abs/1410.6038v3.


 \bibitem{HKM}
Nicholas J A Harvey, David R Karger, Kazuo Murota, ``Deterministic Network Coding by Matrix Completion", in  \textit{Proc. of the 16th Annu. ACM-SIAM symposium on Discrete algorithms, SODA}, pp. 489--498,  2005, DOI: 10.1145/1070432.1070499.
 \bibitem{DSC}
 Son Hoang Dau, Vitaly Skachek and Yeow Meng Chee, ``Error Correction for Index Coding with Side Information'', in \textit{IEEE Trans.  Inf. Theory}, Vol. 59, No. 3, March 2013, pp. 1517-1531.
  \bibitem{LuS}
E. Lubetzky and U. Stav, ``Non-linear Index Coding Outperforming the Linear Optimum", in \textit{Proc. 48th Annu. IEEE Symp. Found. Comput. Sci.}, 2007, pp. 161-168.
 \bibitem{NTZ}
 M. J. Neely, A. S. Tehrani, Z.Zhang, ``Dynamic Index Coding for Wireless Broadcast Networks", in \textit{Proc. 31st  IEEE Conf.  Comput. Commun. (INFOCOM), Orlando, USA.}, March 25-30, 2012, pp. 316-324.
 \bibitem{LaS}
 M. Langberg and A. Sprintson, ``On the Hardness of Approximating the Network Coding Capacity", in \textit{IEEE Symp. Inf. Theory. (ISIT)}, 2008, pp. 315-319.
 \bibitem{RCS}
 S. Y. El Rouayheb, M. A. R. Chaudhry and A. Sprintson, ``On the Minimum Number Of Transmissions in Single-Hop Wireless Coding Networks", in \textit{Proc. Inf. Theory Workshop (ITW), Lake Tahoe},  Sept. 2007.

\end{thebibliography}
\end{document}